\documentclass[final,onefignum,onetabnum]{siamart190516}
\usepackage{graphicx}
\usepackage{color,enumerate}
\usepackage{amsmath,amsfonts,amssymb}
\usepackage[mathscr]{eucal}
\usepackage{bbm}
\usepackage{booktabs}
\usepackage{float}
 \usepackage{hyperref}
 \usepackage[normalem]{ulem}

\def\bbbr{{\rm I\kern-0.23em R}}

\newcommand{\eps}{\varepsilon}
\newcommand{\R}{\mathbb{R}}

\title{
Type III Responses to Transient Inputs in Hybrid Nonlinear Neuron Models}

\author{Jonathan Rubin\footnotemark[2], Justyna Signerska-Rynkowska\footnotemark[3], Jonathan D. Touboul\footnotemark[4]}

\footnotetext[2]{Department of Mathematics, University of Pittsburgh, 301 Thackeray Hall, Pittsburgh, PA, USA\\{(jonrubin@pitt.edu)}}
\footnotetext[3]{Faculty of Applied Physics and Mathematics, Gda\'{n}sk University of Technology, ul. Narutowicza 11/12, 80-299 Gda\'{n}sk, Poland\\ (justyna.signerska@pg.edu.pl)}
\footnotetext[4]{Department of Mathematics and Volen National Center for Complex Systems, Brandeis University, 415 South Street, Waltham, MA, USA\\ (jtouboul@brandeis.edu)}

\headers{J.E.~Rubin, J. Signerska-Rynkowska \& J.D.~Touboul}{Type III Responses in Hybrid Neuron Models}

\begin{document}
\maketitle

\begin{abstract}
Experimental characterization of neuronal dynamics involves recording both of spontaneous activity patterns and of responses to transient and sustained inputs.
While much theoretical attention has been devoted to the spontaneous activity of neurons, less is known about the dynamic mechanisms shaping their responses to transient inputs, although these bear significant physiological relevance. Here, we study responses to transient inputs in a widely used class of neuron models (nonlinear adaptive hybrid models) well-known to reproduce a number of biologically realistic behaviors. We focus on responses to transient inputs that have been previously associated with Type III neurons, arguably the least studied category in Hodgkin's classification, which are those neurons that never exhibit continuous firing in response to sustained excitatory currents. The two phenomena that we study are  \emph{post-inhibitory facilitation}, in which an otherwise subthreshold excitatory input can induce a spike if it is applied with proper timing after an inhibitory pulse, and \emph{slope detection}, in which a neuron spikes to a transient input only when the input's rate of change is in a specific, bounded range. Using dynamical systems theory, we analyze the origin of these phenomena in nonlinear hybrid models. We provide a geometric characterization of dynamical structures associated with PIF in the system and an analytical study of slope detection for tent inputs. While the necessary and sufficient conditions for these behaviors are easily satisfied in neurons with Type III excitability, our proofs are quite general and valid for neurons that do not exhibit Type III excitability as well.  This study therefore provides a framework for the mathematical analysis of these responses to transient inputs associated with Type III neurons in other systems and for advancing our understanding of these systems' computational properties. 
\end{abstract}

\begin{keywords}
Type III excitability, transient responses, hybrid dynamical systems, slope detection, post-inhibitory facilitation.
\end{keywords}

\begin{AMS}
Primary: 37B55,  
34A38 
37C35  
Secondary: 37N25, 
34C60   
\end{AMS}

\section{Introduction}
In the middle of the past century, Hodgkin \cite{hodgkin1948} determined that crustacean axons could be distinguished into classes that are now known as Type I, Type II and Type III, based on their firing properties in response to sustained injected currents (cf. \cite{prescott2008} for a review). In this classification, Type I and Type II neurons fire repeatedly when sufficiently excited but differ in the properties of  the their f-I curves (firing frequency as a function of input amplitude). Type I neurons can respond to small currents with very slow rates; theoretically in this case, the f-I curve is continuous, and as current is decreased from the repeated firing regime, the neuronal firing frequency goes to 0.
 In contrast, Type II neurons cannot maintain arbitrarily slow firing rates, reflected in an f-I curve that undergoes an abrupt jump from 0 (for low currents) to some non-nonzero frequency, which need not be small, as soon as the neuron starts firing. Substantial work has been devoted to the study of these two types of \emph{excitability}; in particular, it was shown early on that in models, these behaviors arise when firing onset occurs via a saddle-node of invariant circles bifurcation or an Andronov-Hopf bifurcation, respectively~\cite{rinzel1998}. 

\newpage

Type III neurons, arguably the least studied class of neurons in this classification, have very distinct properties. In particular, neurons classified as Type III may exhibit transient spiking to current injection, but they will not fire continuously no matter how strong a sustained excitatory current is applied. This property has been considered as advantageous in settings, such as processing of certain auditory stimuli, where each individual spike carries significant meaning or where processing of a rapid stream of stimuli occurs and requires avoidance of overlapping response windows \cite{carr2010,clay2008,izhikevich:07,schnupp2009}.  Type III responses have also been observed in neurons in a variety of additional brain regions including multiple sites in the spinal cord as well as in the neocortex, likely associated with coincidence detection and timing-based coding (see \cite{prescott2008} and the references therein).  Rinzel and collaborators discuss a range of interesting input processing features that they describe as being associated with Type III neurons, including {\em post-inhibitory facilitation (PIF)}, {\em slope detection}, {\em phase locking}, and {\em coincidence detection} \cite{dodla,meng2012}.  PIF refers to the phenomenon in which an excitatory input that fails to induce a spike in a resting neuron can induce a spike when applied with some lag following an inhibitory input.  Note that PIF differs from post-inhibitory rebound, in which a neuron fires immediately upon removal of inhibition.  Slope detection and phase locking are properties associated with repetitive inputs.  Given an input such as a sinusoid, with cycles of rising and falling amplitude, a neuron exhibits slope detection if it only spikes to inputs for which the rate of change is in a specific, bounded range, and it displays phase locking if it only spikes during a certain bounded range of phases within each cycle. Finally, coincidence detection occurs when a neuron responds to two or more inputs if and only if these inputs occur close enough to each other in time.  

In fact, the conditions needed for these properties to arise, and their relationship to Type III responses to applied input, have not yet been established analytically. In this work, we give the first rigorous mathematical treatment of PIF and slope detection. To do so, we consider an accepted planar, hybrid neuronal model that combines continuous evolution of trajectories up to a spiking event, defined by the finite-time blow-up of the voltage variable, together with a discrete jump condition that resets positions of trajectories after spiking occurs.  
By analyzing the hybrid model with a spike threshold and reset, we avoid any ambiguity in what it means for a spike to be fired and we also remove the need to consider a global return mechanism that brings voltage back to baseline levels after it becomes elevated.
 Type III behavior is associated with the existence of a globally stable critical point at resting voltage levels for all levels of input.  For example, in planar systems with continuous vector fields such as the FitzHugh-Nagumo model with variables $(v,w)$ and a cubic nullcline for the voltage variable $v$, Type III behavior can result when the $w$-nullcline is a line at a fixed $v$ value that intersects the resting branch of the $v$-nullcline at all current levels.
Not surprisingly from a dynamical systems point of view, however, our results extend to show that PIF and slope detection do not require a vertical nullcline in the $(v,w)$ plane and can persist even if we vary the model parameters to allow the stable resting critical point to be lost as input increases; that is, we show that Type III responsiveness is not required for PIF and slope detection to arise.
 
The remainder of this paper is organized as follows.  Section~\ref{sec:Model} describes the model we will be studying throughout the paper and summarizes the main results that will be useful in our analysis. Section~\ref{sec:PIF} presents and establishes the dynamical mechanisms supporting PIF in the hybrid neuron model studied, while section~\ref{sec:SlopeDetection} elucidates the mechanisms supporting slope detection. We conclude this paper in section~\ref{sec:Discussion} with a discussion on these results and in particular how these may extend to other types of models.

\section{Neuron Model and Assumptions}\label{sec:Model}
We consider the following modification of the model considered in our previous work \cite{rubinDCDSI,rubinDCDSII}, featuring a voltage variable, $v$, and an adaptation  or recovery variable, $w$:
\begin{equation}
\label{eq:model}
\begin{array}{rcl}
v' & = & F(v)-w+I, \vspace{0.1in} \\
w' & = & bv-cw, 
\end{array}
\end{equation}
where $I$ is the input current.  Often $I$ is considered as a constant but in this paper it will vary over time and may include both positive/excitatory and negative/inhibitory components.
We augment system (\ref{eq:model}) with a discrete reset condition implemented when a trajectory achieves a threshold condition corresponding to spiking, which we discuss further below.
In previous work on models similar to (\ref{eq:model}), the adaptation dynamics has been  given by $w'=\epsilon (bv-w)$, with the explicit timescale parameter $\epsilon$ included to emphasize a possible separation between the timescales of the two variables.  We choose the form given in (\ref{eq:model}) to allow for the case $c=0$ while also allowing the possibility that $w$ could be slow, with $b, c$ small, or not. 
We assume that the parameters and function $F(v)$ in (\ref{eq:model}) satisfy the following assumptions:
\begin{description}
	\item{(A1)} The parameters $b, c$ satisfy $b>0$ and $c\geq 0$.  
	\item{(A2)} $F$ is a convex function with $F'(v)<0$ for $v$ small enough, so that $F$ reaches its global minimum  at a point $(v_f,w_f)$. Moreover, $F(v) > 0$ for all $v>v_0$ for some finite $v_0$, with $\int_{v_0}^{\infty} \frac{dx}{F(x)} < \infty$. 
	
\item{(A3)} For $I=0$, on the left branch of the $v$-nullcline, system (\ref{eq:model}) has a single critical point, which we denote by $(v^*,w^*)$, and it is asymptotically stable. Moreover, system (\ref{eq:model}) does not support a stable periodic orbit.
\end{description} 

By (A2), $v$ will blow up in finite time from some initial conditions.  The blow-up times define firing events and are followed by an instantaneous reset of $v$ and an update of $w$. Often, it is assumed that $F$ grows faster than a quadratic function at infinity in the sense that there exists $\mu>0$ such that $\frac{F(v)}{v^{2+\mu}} \to \infty$ when $v\to \infty$.
In this case, the $w$ remains finite when $v$ blows up~\cite{touboul2009,touboul:09}, and the following reset condition is used:
\begin{equation}
\label{eq:reset}
v(t) \to \infty \; \mbox{as} \; t \uparrow t_0 \; \Rightarrow v(t_0^+) = v_R, w(t_0^+)=w(t_0^-)+w_R
\end{equation}
 for parameters $v_R, w_R$.  
Without this extra assumption on $F$, $w$ may blow up along with $v$, which is problematic for setting a threshold and reset condition~\cite{touboul:09}. For our analysis, the issue of blow-up of $w$ is irrelevant, because the phenomena we consider relate to whether a spike is fired at all, not to what happens after a spike is fired, so the extra assumption on $F$ is not needed.  In our numerics, as in past papers, we shall consider for definiteness the quartic model $F(v)=v^4+\lambda v$ for some $\lambda\in\R$ or a close variant of this model.  

Model (\ref{eq:model}) provides an overarching framework for studying nonlinear adaptive integrate-and-fire neurons, including the classical quadratic~\cite{izhikevich:03,izhikevich:04,izhikevich:07} and exponential~\cite{brette-gerstner:05} models. The dynamics of these systems were studied in detail for $c=1$, and they share a number of common properties.  Notably, their subthreshold dynamics is organized  around a Bogdanov-Takens bifurcation, and they thus all display a saddle-node and a Hopf bifurcation as parameters are varied~\cite{touboul:08}. In particular, when $I$ is large enough, tonic spiking arises when the resting state either loses stability through a  Hopf bifurcation, yielding Type II excitability, or disappears through a saddle-node bifurcation, yielding Type I excitability. 

The dynamics when $c=0$ 
is not topologically equivalent to the case with $c>0$. In particular, the system features a single equilibrium for any value of the input (whereas, for $c>0$, it features between $0$  and $2$ equilibria depending on $I_0$). This equilibrium is given by $v^{*}=0$ and $w^{*}=F(0)+I$. The Jacobian matrix at this fixed point has trace $F'(0)$ and determinant $b>0$, and its stability therefore does not depend on the input parameter. As soon as $F'(0)<0$, the system features a unique, stable fixed point, for any value of the input $I$, a characteristic feature of Type III excitability. In that canonical regime with $c=0$, one could investigate responses to transient inputs when (\ref{eq:reset}) applies with $w_R>0$, through characterizing in particular the attraction basin of the stable fixed point (see~\cite{touboul:09}). Instead of going into this direction, we will instead focus on the dynamical structures that support \emph{type III behaviors} in hybrid systems; we shall see that these will persist for $c\neq 0$. 

\section{Post-Inhibitory Facilitation}\label{sec:PIF}
Our characterization of \emph{post-inhibitory facilitation}  (PIF) follows Dodla et al. \cite{dodla}, who studied PIF for neurons of the medial superior olive and argued that it may be a rather general phenomenon. We characterize PIF as a situation that can arise with a quiescent neuron that fails to spike when it receives a certain excitatory input, but may spike when the same excitatory input is applied after a preliminary inhibitory input. If the excitation follows the inhibition with some delay that is neither too short nor too long, call it $t_e$, then it can cause the neuron to fire, even though the excitation was unable to induce firing on its own; see Figure \ref{fig:PIFintro}.  That is, the inhibition has a {\it facilitatory} role, because it establishes conditions that allow firing that would not have resulted had the inhibition not been applied.  PIF differs from more standard post-inhibitory rebound in that it involves not just the application and removal of inhibition but also requires the involvement of excitation, and has received much less attention in the literature.   

\begin{figure}[htbp]
	\centering
	 \includegraphics[width=5in]{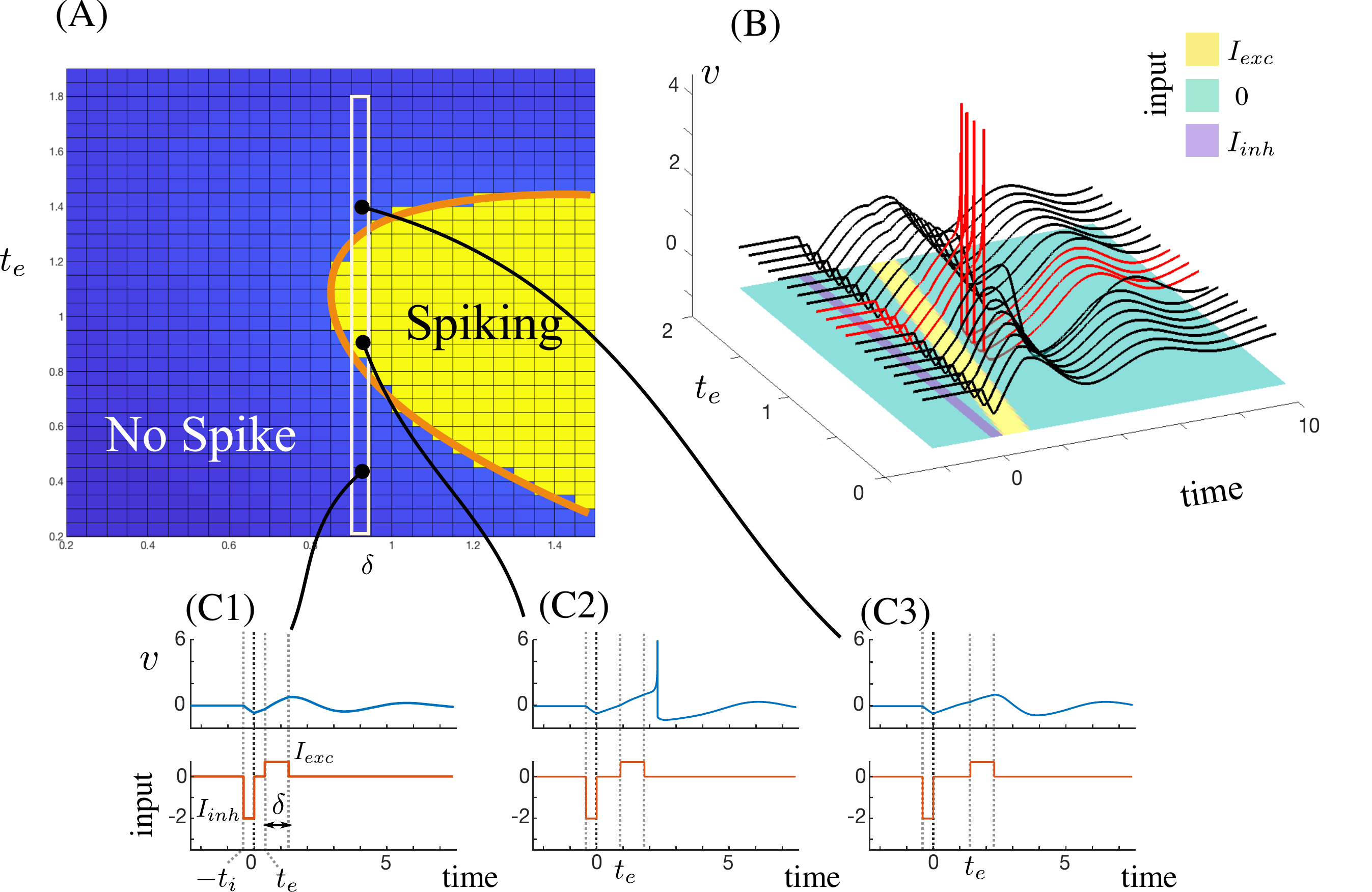}
	\caption{Numerical illustration of the PIF phenomenon. Here $F(v)=v^4-0.5v, b=2$, and $c=0$.  From the initial condition $(v,w)=(0,0)$, which is a critical point of (\ref{eq:model}) with $I=0$, an inhibitory input with $I_{inh}= - 2$ is applied for $t\in[-t_i,0]$ with $t_i=0.4$ time units.  Excitation $I_{exc}= 0.7$ is applied at time $t_e$ and kept on for time $\delta$.  (A):  Spiking occurs only for some choices of $t_e, \delta$, and never for $t_e=0$. (B):  Another view of the PIF phenomenon.  Here $\delta = 0.9$ is fixed.  The time when inhibition is on is labeled in purple, while the time when excitation is on is shown in yellow.  Spiking occurs for $t_e \in (0.8,1.3)$ (red voltage traces) but not for other choices of $t_e$ (black voltage traces). (C1),(C2),(C3): Example voltage traces for points in (A) with $\delta=0.9$ and $t_e \in \{ 0.4, 0.9, 1.4 \}$.}
	\label{fig:PIFintro}
\end{figure}

To study the mathematical structures that may support the PIF phenomenon, we will consider trajectories of system~\eqref{eq:model} with a time-varying input $I$ that is piecewise constant, composed of an inhibition phase ($I<0$) followed by a resting phase ($I=0$) and by an excitation phase ($I>0$). Assuming that the neuron does not spike in response to this stimulation, the resulting trajectories will be continuous but not smooth, formed as the concatenation of several segments:  an initial segment from $(v^*,w^*)$ defined with inhibition on, corresponding to $I<0$, which we call $\Phi_i(t)=(v_i(t),w_i(t))$; a second segment defined with $I=0$, which we call $\Phi_0(t) = (v_0(t),w_0(t))$; and a third segment defined with excitation on, corresponding to $I>0$, which we call $\Phi_e(t) = (v_e(t),w_e(t))$.  



We work under assumptions (A1)-(A3), that ensure simple structural conditions about the nullclines and the direction of the vector field.  (A2)-(A3) always hold for the exponential and the quadratic models, and also for the quartic model $F(v)=v^4+\lambda v$ for appropriate choices of parameters
  (see also Figure \ref{fig:setup}).
Let us now make some useful remarks derived from the current assumptions. From assumption (A1), we observe that the $w$-nullcline has positive (or infinite if $c=0$) slope, and in particular that for each fixed $w$ value, if we let $v(w)=cw/b$, then 
we have $w'>0$ for $v>v(w)$ and $w'<0$  for $v<v(w)$; in the case of infinite slope ($c=0$), the value of the voltage $v(w)$ is  actually independent of $w$. 
From the convexity of $F$ in assumption (A2), we notice that considering only the dynamics of the $v$-equation from (\ref{eq:model}) with $w>w_f$ fixed, the left branch of the $v$-nullcline  is attracting and the right branch is repelling. We shall denote the left nullcline branch by $\{ (V^-(w),w) \}$ and the right branch by $\{ (V^+(w),w) \}$ (see Figure~\ref{fig:setup}).
When the $w$-nullcline has finite positive slope, in addition to the  stable fixed point of Assumption (A3), there will be a single critical point on the right branch of the $v$-nullcline, which is a saddle, which we will denote $(v_r,w_r)$, with $w_r>w^*$. No such point exists when $c=0$. 

To characterize the emergence of PIF in the hybrid neuron model, our analysis relies on the concept of a {\it firing threshold curve}~\cite{touboul2008}. This concept is close to the notion of {\it spike threshold} used broadly in neuroscience and ties in with the concept of excitability.  The idea of the threshold is that for each neuron, there is a voltage level such that the neuron will generate an action potential if and only if its membrane potential exceeds that level.  Sometimes a similar idea is referenced as a  current threshold, such that a neuron at rest will spike if and only if an applied current exceeds that level. Computational modeling shows that a voltage threshold should not be considered as a fixed value for a given neuron, but rather depends on the levels of other quantities associated with that neuron, such as the activation and inactivation levels of its voltage-gated currents, and some mathematical work has considered how to  precisely identify this more complicated structure \cite{letson2018,mitry2013,plat2010,plat2011}. Given that we define spiking for system (\ref{eq:model}) based on finite time blow up of trajectories (in $v$, while $w$ remains bounded if $F(v)$ grows faster than quadratic in $v$), we define the firing threshold curve for system (\ref{eq:model}) as a curve that separates trajectories that blow up in finite time from those that remain bounded.

\begin{figure}[htbp]
	\centering
	\includegraphics[width=\textwidth]{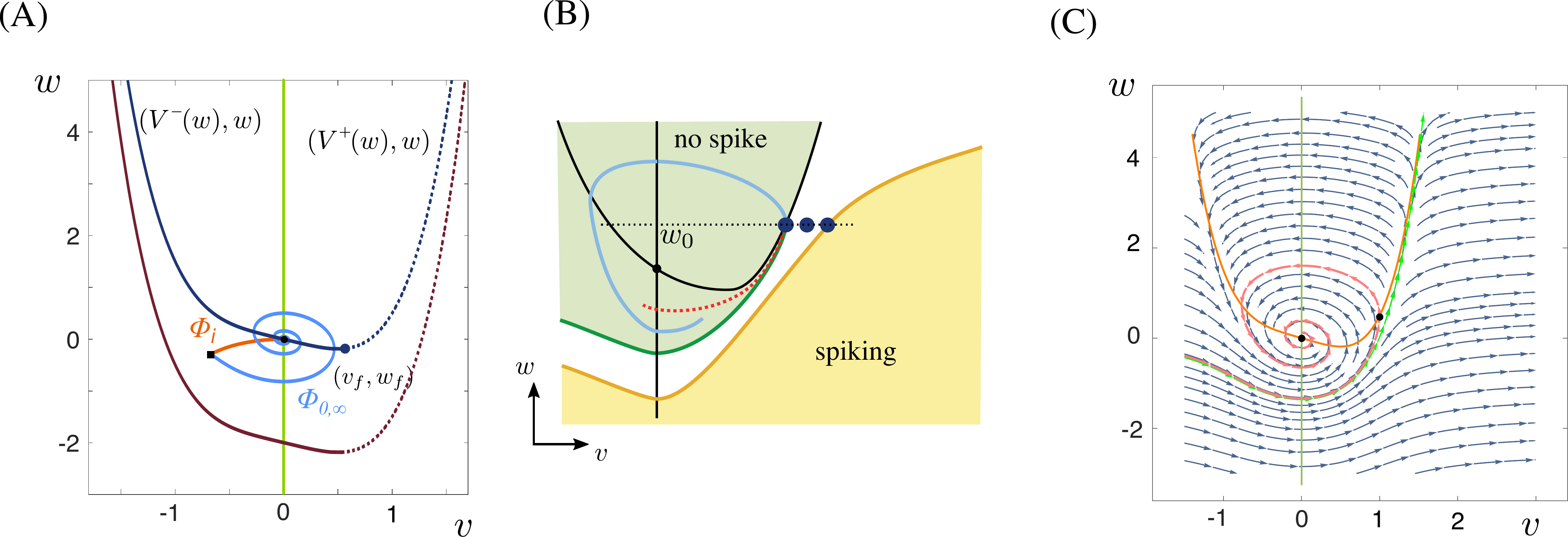}
	\caption{Phase plane for system (\ref{eq:model}) with $F(v)=v^4-0.5v, b=2, c=0$. (A) Response to a transient step of inhibitory current of amplitude $I_{inh}=-2$ applied in time interval $[0,t_i]$ with $t_i=0.4$. Diagram shows the $v$-nullcline for $I=0$ (dark blue, solid: stable part, dotted: unstable part, circle: fold point) and $I=-2.0$ (dark red), $w$-nullcline (green) and the continuous trajectory formed by concatenating $\Phi_i(t)$ (orange) that starts from $(v^*,w^*)$ at $t=0$ (black circle) and ends at $\Phi_0(0) \equiv (v_i(t_i),w_i(t_i))$ (black square, $t_i=0.4$), and $\Phi_{0,\infty}(t)$ (light blue) which starts from $\Phi_0(0)$ and tends to $(v^*,w^*)$ as $t \to \infty$. Later, we will refer to this concatenated curve as ${\cal C}$. (B) Existence of a threshold. Sketch of a trajectory through a point from $(V^+(w_0),w_0)$ for some $w_0>w_f$, with in black the nullclines, in blue the forward trajectory, in dotted red an impossible backward trajectory and in green a typical backward trajectory from this point (see text). The concatenation of that green trajectory with the branch $\{(V^+(w),w),w>w_0\}$ delineates a non-spiking region. In contrast, for $v_0$ large enough the neuron fires a spike, and below this orbit all trajectories spike (yellow region). The threshold lies between the green and yellow region.  (C) Stream plot of the vector field (blue vector) with a specific trajectory (green, through $(v,w)=(0, -1.363177138)$) approximately matching the asymptotic river~\cite{letson2018} and splitting the phase space intro trajectories going to $(v^*,w^*)$ and those associated with a spike. Pink trajectory through a point on $(V^+(w),w)$ (black circle), converges forward to $(v^*,w^*)$, and diverges backwards. } 
	\label{fig:setup}
\end{figure}

The details of the firing threshold curve for system (\ref{eq:model}) depend on whether $c$ is nonzero and on the form of $F(v)$ \cite{touboul2009}.  
When $c>0$, the fixed point $(v_r,w_r)$ is a saddle and the one-dimensional stable manifold of $(v_r,w_r)$ (or a part of it) may form a separatrix that acts as a firing threshold curve (see Fig. 5b and 5c in \cite{touboul2009}), provided that $(v^*,w^*)$ is not an unstable focus, which is excluded by assuming that $(v^*,w^*)$ is stable. This stable manifold acts as the firing threshold unless there is an unstable periodic orbit surrounding the stable fixed point $(v^*,w^*)$, in which case that orbit becomes the firing threshold curve, and one branch of the stable manifold wraps around this periodic orbit (cf. \cite[Fig. 5a]{touboul2009} or \cite{rubinDCDSII}). This case remains as a possibility when $c=0$. However, when $c=0$ and there is no unstable periodic orbit, since the saddle point $(v_r,w_r)$ no longer exists, a new approach is needed to define a firing threshold curve that splits the phase plane into two regions, one where trajectories converge (forward in time) to $(v^*,w^*)$ and the other where the solutions blow up. To characterize the existence of this threshold, we will show that for any $w_0>w_f$ (the minimum of $w$ along the $v$-nullcline), there exists a voltage $v_{th}(w_0)$ and an $\eps>0$ such that for any $v_0>v_{th}(w_0)$, the voltage blows up and for $v_0 \in [v_{th}(w_0)-\eps,v_{th}(w_0))$ the voltage does not blow up. The full orbit through $(v_{th}(w_0),w_0)$ forms the threshold. Fix $w_0>w_f$. It is easy to show that for $v_0$ large enough, say $v_0=\bar{v}_0(w_0)$, the solution will blow up (see, e.g., trapping regions for spiking outlined in~\cite{touboul2008,touboul:09}), and the part of the phase space below such trajectories is all associated with spiking (Figure~\ref{fig:setup}B, yellow region). Showing the existence of a threshold thus amounts to finding a value of the voltage for which the solution does not blow up, as illustrated in Figure~\ref{fig:setup}B. In fact, it suffices to notice that the trajectory of~\eqref{eq:model} with initial condition $(V^+(w_0),w_0)$ converges to the stable critical point $(v^*,w^*)$ as $t \to \infty$. Figure~\ref{fig:setup}B illlustrates why this convergence must hold. At the initial condition, the vector field points vertically up into the green region, the forward trajectory (blue) goes up and to the left, loops around the fixed point, crosses downwards through $(V^-(w),w)$ and then crosses the $w$-nullcline. The backward trajectory from $(V^+(w_0),w_0)$ will go down and left and cross the $w$-nullcline below the fixed point. This crossing cannot occur above forward trajectory crossing (Fig.~\ref{fig:setup}B, dashed red trajectory), since that would yield a bounded trapping region for the backward trajectory but there is no critical point or periodic orbit to which it could converge as $t \to -\infty$. The backward trajectory therefore crosses the $w$-nullcline below the forward trajectory (Fig.~\ref{fig:setup}B, green trajectory) and diverges as $t\to -\infty$, while trapping the forward orbit, forcing it to converge to $(v^*,w^*)$. Therefore, the region of the phase space above the concatenation of $\{(V^+(w),w), w>w_0\}$ and the backward orbit associated with initial condition $(V^+(w_0),w_0)$ (light green in Figure~\ref{fig:setup}B) is composed of trajectories converging to $(v^*,w^*)$.   Noting that $v'$ is increasing in $v$ to the right of $\{ (V^+(w),w)\}$, we obtain a unique threshold value $v_{th} = \inf \{ v_0 >V^+(w_0): \;  \mbox{the solution } \; (v(t),w(t)) \; \mbox{with initial condition}  \; (v_0,w_0)  \;\mbox{satisfies} \; v(t) \to \infty \; \mbox{as} \; t \to \infty \}$ with $v_{th} \in (V^+(w),\bar{v}_0(w))]$, and the trajectory through $(v_{th},w_0)$ serves as the firing threshold curve.  We can define a similar curve for each fixed $I$, and we denote each such curve by ${\cal F}(I)$. 
Note that because $w$ is monotone increasing to the right of the $w$-nullcline, each ${\cal F}(I)$ can be represented as the graph of an increasing function of $v$ there.

To approximate this firing threshold numerically, we use the recently defined concept of an {\em asymptotic river} in planar systems~\cite{letson2018} (green orbit in Figure~\ref{fig:setup}C).  To define an asymptotic river, we can examine the locus of zero curvature and the curve of zero torsion (i.e., zero derivative of curvature) for the flow of system (\ref{eq:model}) and determine where these curves converge together as they approach infinity.  By setting a condition on the closeness of these curves within a finite region of phase space, we can find a point that is arbitrarily close to lying on an asymptotic river, and we can use the trajectory through that point as an approximation to an asymptotic river. Using this approach, we find that an approximate asymptotic river for system (\ref{eq:model}) with fixed $I$, which apparently converges to the right branch of its $v$-nullcline as $t \to \infty$ (Fig.~\ref{fig:setup}C), provides an excellent approximation to its firing threshold (see also \cite{desroches2013}), and thus we use such a trajectory in the  numerical illustrations with $c=0$ that follow. 

For fixed $I$ for which system (\ref{eq:model}) has a stable equilibrium point, say $(v^*(I),w^*(I))$ with $(v^*(0),w^*(0))=(v^*,w^*)$ from (A3), denote the basin of attraction of $(v^*(I),w^*(I))$ by ${\cal A}(I)$ and the set of initial conditions that blows up in finite time by ${\cal B}(I)$.  The curve ${\cal F}(I)$ forms part of the boundary between these sets. That is, any neighborhood of a point on  ${\cal F}(I)$ will intersect both ${\cal A}(I)$ and ${\cal B}(I)$.  We will use these sets in our analysis.


Next, we will introduce additional notation (see Figure \ref{fig:setup}A).    Fix any level of inhibitory input $I_{inh}<0$.  For $t_i>0$ sufficiently small, the solution of system (\ref{eq:model}) with $I=I_{inh}$ and with initial condition $(v^*,w^*)$, which we denote by $\Phi_i(t) = (v_i(t),w_i(t))$ with $\Phi_i(0)=(v^*,w^*)$,  satisfies $v_i'(t)<0$ for all $t \in [0,t_i]$.   
For PIF, 
we are interested in the trajectory formed by concatenating $\Phi_i(t)$ together with segments $\Phi_0(t)$ and $\Phi_e(t)$, each defined on a finite time interval, as mentioned earlier; to actually generate a spike, we will also need a final segment -- from the flow with $I=0$ -- concatenated after $\Phi_e(t)$, so we will include that segment in our definitions but it will not be critical in our analysis.
We fix the initial trajectory segment $\Phi_i(t)$   on the time interval $[0,t_i]$  and consider a trajectory defined forward in time from the termination point of $\Phi_i(t)$, namely $(v_i(t_i),w_i(t_i))$.
For convenience, we introduce a shifted time variable, so that the trajectory we define starts from $t=0$ at $(v_i(t_i),w_i(t_i))$.
To generate this trajectory, we use system (\ref{eq:model}) with piecewise constant $I$ that depends on several positive parameters, which we leave unspecified for now.
Specifically, for any $I_{exc}, \delta, t_e > 0$, define
\begin{equation}
\label{eq:pieceI}
I(t) = \left\{ \begin{array}{l} 0, \; t \in [0,t_e), \vspace{0.1in} \\ I_{exc}, \; t \in [t_e,t_e+\delta), \vspace{0.1in} \\  0, \;  t \geq t_e+\delta. \end{array} \right. 
\end{equation}
Let $\Phi(t;I_{exc},t_e,\delta)$ denote the solution of (\ref{eq:model}) with initial condition $(v_i(t_i),w_i(t_i))$ and with $I(t)$ given by (\ref{eq:pieceI}).  We can now give a more mathematically precise definition of PIF; see also Figure \ref{fig:PIF}.  

\medskip

\noindent {\bf Definition.}  Fix the inhibition strength $I_{inh}<0$ and choose the inhibition duration $t_i$ to be a time (dependent on $I_{inh}$) such that for all $t \in (0,t_i)$, we have both $v_i'(t)<0$ and $\Phi_i(t) \in {\cal A}(0)$. PIF occurs for system \eqref{eq:model} if assumptions (A1)-(A3) hold and there exists an interval of positive real numbers ${\cal I}  = (\underline{I}_{exc},\overline{I}_{exc})$, 
such that:
\begin{enumerate}
\item if $I_{exc} \in {\cal I}$, then there exist:
\begin{description}
\item{(a)} an interval of positive, finite real numbers ${\cal T}(I_{exc}) = (\underline{t}_e(I_{exc}),\overline{t}_e(I_{exc}))$ and,  
\item{(b)} for each  $t_e \in {\cal T}(I_{exc})$, a corresponding constant $\underline{\delta}(I_{exc},t_e)  >  0$ 
\end{description}
such that  $\Phi(t;I_{exc},t_e,\delta)$ yields a spike if and only if $t_e \in {\cal T}(I_{exc})$ and $\delta > \underline{\delta}(I_{exc},t_e)$,  
\item if $I_{exc} < \underline{I}_{exc}$, then for all positive $t_e$ and $\delta$, $\Phi(t;I_{exc},t_e,\delta)$ does not yield a spike, and
\item if $I_{exc} > \overline{I}_{exc}$, then there exists $\delta > 0$ such that $\Phi(t;I_{exc},0,\delta)$ yields a spike.
\end{enumerate}

\medskip

\begin{figure}[htbp]
	\includegraphics[width=\textwidth]{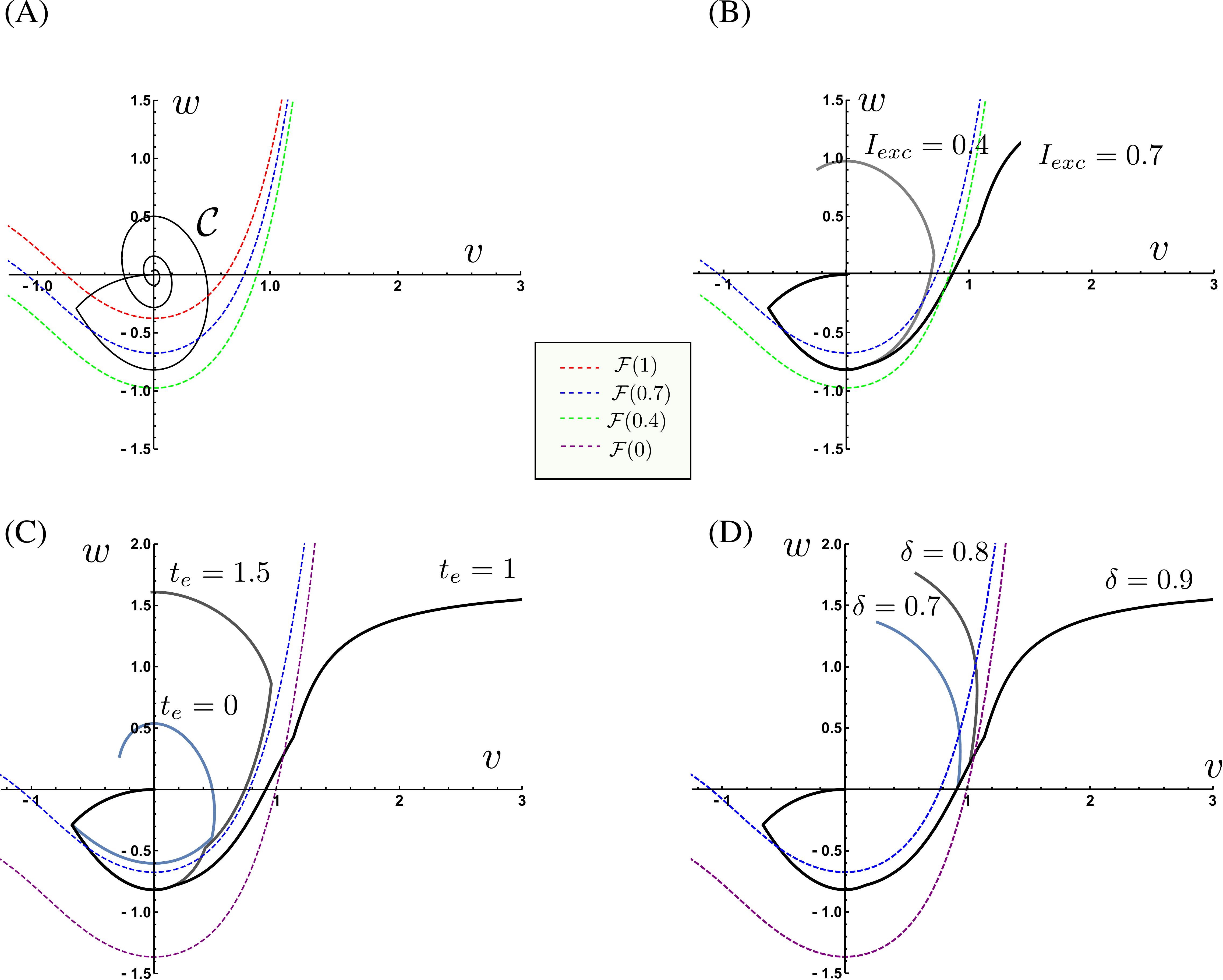}  
 %
\begin{caption}{
Illustration of the conditions for PIF with $F(v)=v^4-0.5v$, $b=2$, and $c=0$.  In all cases, the critical point in the absence of input lies at the origin. Colored dashed curves indicate ${\cal F}(I)$ for $I=0$ (purple), $I=0.4$ (green), $I=0.7$ (blue), and $I=1.0$ (red). The inhibition applied in all cases is of amplitude $I_{inh}=-2$ and duration $t_i=0.4$, followed by an excitation of different amplitudes $I_{exc}$ (B), delay $t_e$(C) or duration $\delta$. Unless otherwise indicated, $I_{exc}=0.7$ $t_e=1$, $\delta=0.9$.
(A):  The concatenated trajectory ${\cal C}$ (black) in response to the inhibition only. Note that ${\cal C}$ lies above (i.e., on the non-spiking side of) ${\cal F}(0.4)$, and therefore any excitation with $I_{exc}=0.4$ will not produce a spike. Part of ${\cal C}$ lies below ${\cal F}(0.7)$, but the leftmost point, $(v_i(t_i),w_i(t_i))$, does not. Since $(v_i(t_i),w_i(t_i))$ lies below ${\cal F}(1.0)$, the definition of PIF implies that PIF cannot occur with $I_{exc}=1.0$. $I_{exc}=0.7$ is thus a reasonable excitation value for PIF. 
(B):  Role of $I_{exc}$. When excitation of amplitude $I_{exc}=0.4$ (grey) or $I_{exc}=0.7$ (black) is applied, ${\cal C}$ is below ${\cal F}(I=0.7)$ but above ${\cal F}(I=0.4)$.  Hence, only the stronger excitation can result in a spike.  
(C):  Role of $t_e$. For small ($t_e=0$, left blue-grey trajectory) or large ($t_e=1.5$, grey trajectory) delay, the trajectory remains above 
${\cal F}(0)$ (dashed purple) and therefore return to rest, while for intermediate delay ($t_e=1$, black trajectory), the trajectory crosses ${\cal F}(0)$ (purple) and spikes.  
(D): Role of $\delta$: Trajectories associated with excitation durations $\delta=0.7$ (left, blue-grey), $0.8$ (center, grey), or $0.9$ (right, black). Only a long enough excitation allows the trajectory to cross ${\cal F}(0)$ (purple) and results in a spike.
 } \label{fig:PIF}
\end{caption}
\end{figure}

In other words, PIF describes the situation in which the application of excitatory inputs within a certain bounded range of magnitudes can induce a spike if and only if the excitation is introduced within an appropriate, bounded (since $\overline{t}_e(I_{exc})$ is finite) time window after the offset of inhibition.  On the other hand, if excitation is too weak, if it comes on or turns off too early, or if it comes on too late, then a spike is not fired.   There is a subtle point in this definition:  because $\Phi_i(t_i) \equiv  (v_i(t_i),w_i(t_i)) \in \mathcal{A}(0)$, it follows that in the absence of excitation, the solution to (\ref{eq:model}) with $I=0$ from $(v_i(t_i),w_i(t_i))$ will converge back to $(v^*,w^*)$ as $t \to \infty$ (Fig. \ref{fig:setup}).  Thus, the inclusion of the condition that $\overline{t}_e$ is finite for each $I_{exc} \in {\cal I}$ implies that the solution to (\ref{eq:model}) with $I=I_{exc}$ and initial condition $(v^*,w^*)$ will not produce a spike.    
Therefore, the application of inhibition is crucial for allowing the spike to occur, as desired.
We emphasize that the condition $(v_i(t_i),w_i(t_i)) \in \mathcal{A}(0)$ relates the position of $\Phi_i(t)$, which is defined from system (\ref{eq:model}) with $I=I_{inh}$, relative to the curve ${\cal F}(0)$ and the set ${\cal A}(0)$, which are defined from (\ref{eq:model}) with $I=0$. 
Another subtle point is that all of the quantities appearing in the definition in general will depend on the strength $I_{inh}$ and the duration $t_i$ of the inhibition that is applied before the excitation.

PIF may seem like a specialized property, and indeed it would be if $I_{exc}$, $t_e$, and $\delta$ were pre-specified.  But in fact, since the definition of  PIF allows flexibility in choosing $I_{exc}, t_e$, and $\delta$, we do not need elaborate conditions to ensure that PIF occurs.  Indeed, we have the following main result on PIF:


\medskip

\begin{theorem} \label{thm:PIF}
 Suppose that assumptions (A1)-(A3) hold for system (\ref{eq:model}).  Fix $I_{inh}, t_i$ as in the defintion of PIF.  Define  $\Phi_{0,\infty}(t) = (v_0(t),w_0(t))$ as the solution to system (\ref{eq:model}) with $I=0$ on the time interval $t \in [0,\infty)$ with initial condition $\Phi_{0,\infty}(0)=(v_i(t_i),w_i(t_i))$.  If $\Phi_{0,\infty}(t)$ remains bounded for all $t>0$ and lies in the basin of attraction of $(v^*,w^*)$, then PIF occurs for system (\ref{eq:model}) for this $I_{inh}, t_i$ (and hence for open intervals containing $I_{inh}, t_i$). 
 \end{theorem}
\medskip

\noindent
{\bf Remark:}  Theorem 1 shows that PIF is a quite general phenomenon, but it does not address robustness.  We can think of the size of ${\cal I}, {\cal T}$ as indicators of how robust PIF is for a given system and parameter set.  The  robustness of PIF then depends on a variety of factors.  Two factors are the size of the basin of attraction of $(v^*,w^*)$ with $I=0$, ${\cal A}(0)$, and the distance of $(v^*,w^*)$ from the boundary $\partial {\cal A}(0)$.  If a large inhibition can be applied without causing $\Phi_i$ to cross $\partial {\cal A}(0)$, then that allows $\Phi_{0,\infty}$ to deviate relatively far from $(v^*,w^*)$ as it converges to  $(v^*,w^*)$, which provides an opportunity for subsequent excitation to induce a spike.  Two additional, not entirely independent factors are the size of $b, c$ in system (\ref{eq:model}) and the strength with which $(v^*,w^*)$ attracts trajectories, which depends on the eigenvalues of the Jacobian of system (\ref{eq:model}) at $(v^*,w^*)$.  With weaker attraction to $(v^*,w^*)$, trajectories can undergo larger excursions in their approach to $(v^*,w^*)$, which also favors PIF.

\medskip

\begin{proof}
We start by fixing $I_{inh}, t_i$ as needed\footnote{In fact, for any $I_{inh}<0$, there exists a suitable choice of $t_i$ as required in the definition of PIF.}. Consider the curve ${\cal C}$ in the $(v,w)$ plane constructed as the concatenation of $\Phi_i(t)$ together with the trajectory $\Phi_{0,\infty}(t)$ (see Figure \ref{fig:setup}A or Figure \ref{fig:PIF}A).  ${\cal C}$ is a continuous, closed and bounded, and includes $(v^*,w^*)$.  
Regardless of whether $(v^*,w^*)$ is a node or a focus, ${\cal C}$ will achieve its minimum value in $w$ at the point where $\Phi_{0,\infty}(t)$ first intersects the $w$-nullcline, which we denote by $\Phi_{0,\infty}(t^-) = (v^-,w^-) \equiv (cw^-/b,w^-)$,  and it will achieve its maximum value in $v$ at the point where $\Phi_{0,\infty}(t)$ first intersects the $v$-nullcline, which we denote by $\Phi_{0,\infty}(t^+) = (v^+,w^+)$.

Recall that the firing threshold curve ${\cal F}(I)$ is defined for each $I$.
For any fixed $I>0$, we can compare the location of ${\cal C}$ to the firing threshold curve ${\cal F}(I)$ for that $I$.  When we consider the flow of (\ref{eq:model}) with that $I$ value, we refer to  a point as lying {\em above} ${\cal F}(I)$ (or belonging to ${\cal A}(I)$) if the trajectory emanating from that point does not blow up in finite time and as lying {\em below} ${\cal F}(I)$ (or belonging to ${\cal B}(I)$) if the trajectory emanating from that point does blow up in finite time.  
Necessary and sufficient conditions for a given choice of $I_{exc}>0$ to belong to an interval ${\cal I}$ for which PIF occurs are (see Figures \ref{fig:PIF} and \ref{fig:intersect}):
\begin{description}
\item{(C1)} $(v_i(t_i),w_i(t_i))$ lies above ${\cal F}(I_{exc})$.  This condition implies that there is a window of time after the offset of inhibition when application of excitation will not cause firing.
\item{(C2)} There exists a point on ${\cal C}$ that lies below ${\cal F}(I_{exc})$. This condition implies that there is a positive time such that if excitation is turned on at that time, then it can cause firing if it is left on long enough.
\item{(C3)} $(v^*,w^*)$ lies above ${\cal F}(I_{exc})$.  This condition implies that if excitation is applied too late, then firing will not result.  
\end{description}


Our goal is to show that there exists a bounded interval of positive $I_{exc}$ values for which conditions (C1)-(C3) hold.
Since $\Phi_{0,\infty}(t) \to (v^*,w^*)$ as $t \to \infty$, ${\cal C}$ lies above ${\cal F}(0)$.
By comparison of vector fields, if the trajectory from an initial condition $(v_0,w_0)$ blows up in finite time with $I=I_1>0$, then the trajectory from $(v_0,w_0)$ also blows up in finite time with $I=I_2>I_1$.
Thus, all points below ${\cal F}(I_1)$ also lie below ${\cal F}(I_2)$, and in general, there is an ordering of the threshold curves ${\cal F}(I)$, and the spiking regions ${\cal B}(I)$ form an increasing sequence in the sense that $I_1 < I_2$ implies ${\cal B}(I_1) \subset {\cal B}(I_2)$ 
(Figure \ref{fig:intersect}). 
For each $I$, the corresponding threshold curve ${\cal F}(I)$ lies to the right of the corresponding $v$-nullcline as $t \to \infty$.   Recalling that ${\cal F}(I)$ is a trajectory and following it backwards in time, it will progress in the direction of decreasing $w$ until it eventually crosses the $w$-nullcline and proceeds in the direction of increasing $w$.  After this crossing, it may or may not intersect the $v$-nullcline, depending on $F(v)$. 
If not, then we can represent ${\cal F}(I)$ by $\{ v,W(v,I) \}$, with $\partial W(v,I) / \partial I > 0$. 
If at least one intersection with the $v$-nullcline does occur, then 
we can represent the part of ${\cal F}(I)$ defined for times greater than the largest-time such intersection  (i.e., the first one that occurs as we follow ${\cal F}(I)$ backwards in time) by $\{ v,W(v,I) \}$, still with $\partial W(v,I) / \partial I > 0$. 

\begin{figure}[htbp]
	\begin{center} 
		\hspace{1in} \includegraphics[width=4in]{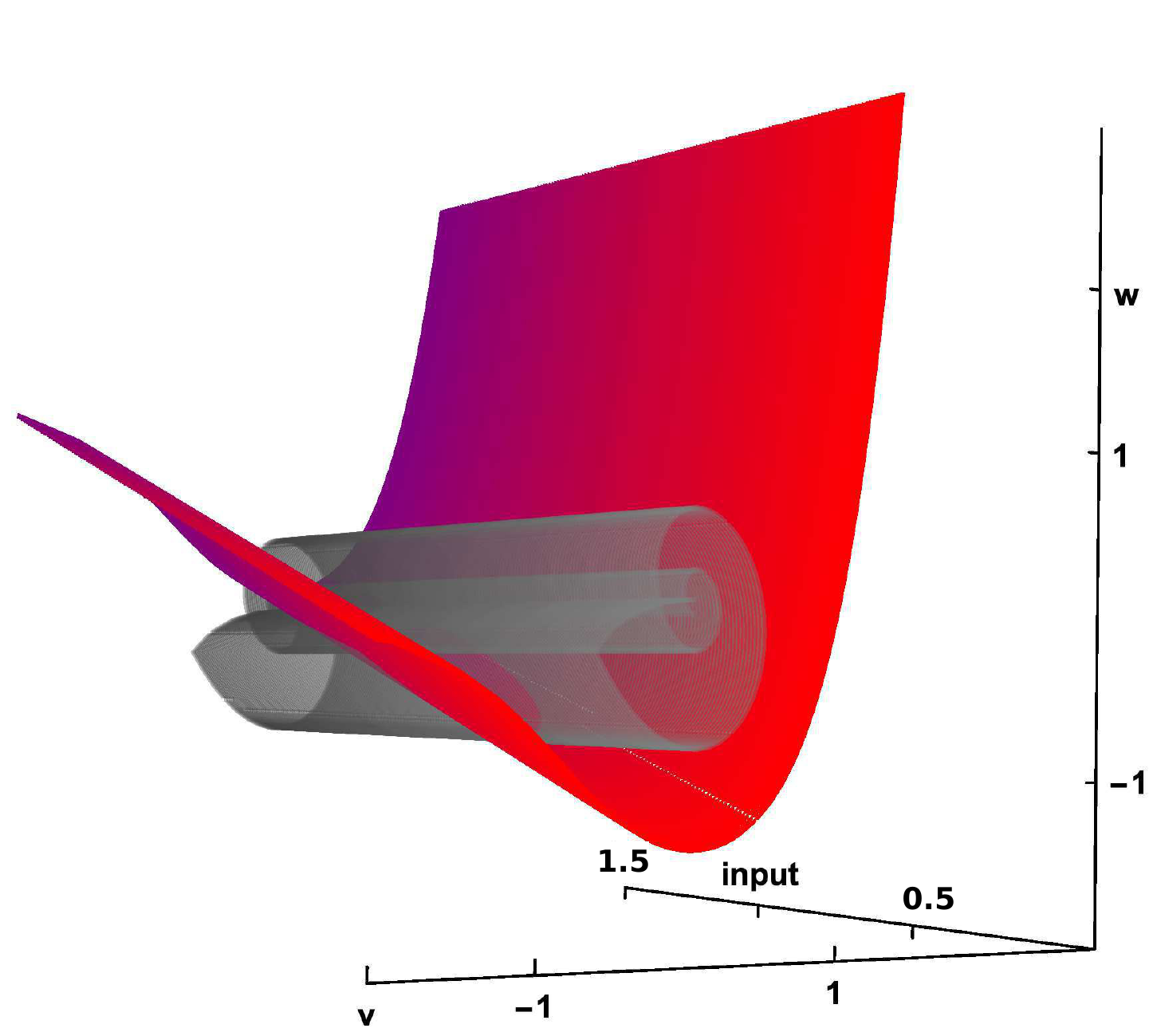} 
	\end{center}
	
	\caption{For any fixed range $J$ of $I$ values, we can form a topological cylinder in $(v,w,I)$-space as ${\cal C} \times J$.  
Abusing notation, define the two-dimensional firing surface ${\cal F} = \{ (v,w,I) : w=W(v,I) \} = \cup_{I} {\cal F}(I)$ in three-dimensional $(v,w,I)$ space (with $\dot{I}=0$). 
We can visualize this object (red/purple) together with ${\cal C} \times J$ (grey) in a 3-d view.  In this image, made using the same parameters as Figure \ref{fig:PIF}, $J = [0,1.5]$, labeled as ``input'' in the figure.  For sufficiently small input $I$, ${\cal C}$ lies above ${\cal F}(I)$. For some $\underline{I}_{exc}>0$,  ${\cal C} \times \underline{I}_{exc}$ first achieves a point of tangency with ${\cal F}$.  From there, continuing to increase $I$, we obtain a candidate for $\overline{I}_{exc}$ at the minimal $I$ value where $(v_i(t_i),w_i(t_i))$ intersects ${\cal F}$ (visible toward the left side of the plot).}
	\label{fig:intersect}
\end{figure}

Now we show that for $I$ sufficiently large, the point $(v^-,w^-) \in {\cal C}$ lies below ${\cal F}(I)$.
Because $W(v,I)$ is increasing in $I$, it suffices to show that for $I$ sufficiently large, the trajectory of (\ref{eq:model}) with initial condition $(cw^-/b,w^-)$, call it $\Phi((cw^-/b,w^-),t;I)$, crosses ${\cal F}(0)$, in which case it will certainly lie below ${\cal F}(I)$.
Pick a point $(v_t,w_t) \in {\cal F}(0)$ with $v_t > cw^-/b$ and $w_t > w^-$.   
We will attain the desired result by showing that $\Phi((cw^-/b,w^-),t;I)$ reaches $\{ v = v_t \}$ before it reaches $\{ w = w_t \}$. 
Define $k:= \min \{ F(v)-w : v \in [cw^-/b,v_t], w \in [w^-,w_t] \}$.
An upper bound $\overline{t}$ on the time for $\Phi((cw^-/b,w^-),t;I)$ to reach $\{ v = v_t \}$ is given by solving $v' = k + I, v(0)=cw^-/b, v(\overline{t}) = v_t$ to obtain
\[
\overline{t} = \frac{\textstyle v_t - cw^-/b}{\textstyle k+I}.
\]
A lower bound $\underline{t}$ on the time for $\Phi((cw^-/b,w^-),t;I)$ to reach $\{ w = w_t \}$ is given by solving $w' = bv_t - cw^-$ with conditions $w(0)=w^-, w(\underline{t}) = w_t$ to obtain
\[
\underline{t} = \frac{\textstyle w_t - w^-}{\textstyle bv_t - cw^-}.
\]
For $I$ sufficiently large, $\overline{t} < \underline{t}$, and hence $\Phi((cw^-/b,w^-),t;I)$ reaches $\{ v = v_t \}$ before it reaches 
$\{ w = w_t \}$, which implies that $\Phi((cw^-/b,w^-),t;I)$ lies to the right of ${\cal F}(0)$, and hence to the right of ${\cal F}(I)$, when it reaches $\{ w = w_t \}$.

As a consequence, (C2) holds for $I_{exc}$ sufficiently large.  Correspondingly, as $I$ is raised from 0, there will be a unique positive $I$ value, which we can take as $\underline{I}_{exc}$, where ${\cal C}$ first achieves one or more points of tangency with ${\cal F}(I)$ (Figure \ref{fig:intersect}). To see that (C1)-(C3) hold, it remains to show that the set of tangent points with $I =  \underline{I}_{exc}$ excludes $(v_i(t_i),w_i(t_i))$ and $(v^*,w^*)$.

%
By way of contradiction, suppose that the initial set of points of tangency includes $(v_i(t_i),w_i(t_i))$.  Then there exists $I>\underline{I}_{exc}$  such that $(v_i(t_i),w_i(t_i)) \in {\cal B}(I)$ and there exists $t>0$ such that for that specific $t$ value, $\Phi_{0,\infty}(t):=(v_{0,\infty}(t),w_{0,\infty}(t)) \in {\cal A}(I)$ with $w_{0,\infty}(t) < w_i(t_i)$.
The trajectory from $(v_i(t_i),w_i(t_i))$ for this $I$ value must lie above $\Phi_{0,\infty}(t)$, since the latter is defined from system (\ref{eq:model}) with $I=0$, whereas $I>\underline{I}_{exc}>0$.
But that means that this trajectory must cross ${\cal F}(I)$ and enter ${\cal A}(I)$, such that it remains bounded as $t \to \infty$, which contradicts the assumption that $(v_i(t_i),w_i(t_i)) \in {\cal B}(I)$.  Thus, the first tangency occurs away from $(v_i(t_i),w_i(t_i))$. 
Similarly, the first tangency also cannot occur at $(v^*,w^*)$.
Indeed, $(v^*,w^*)$ lies on the $w$-nullcline.  Hence, if ${\cal F}(I)$ passes through $(v^*,w^*)$, then all other points that lie on ${\cal F}(I)$ below the $v$-nullcline for that $I$ have $w>w^*$, whereas the entire segment $\Phi_i \subset {\cal C}$ consists of points with $w< w^*$, such that ${\cal F}(I)$ must have already crossed through this segment to achieve a tangency with $(v^*,w^*)$.  Thus, there exists an interval of $I$ values of the form $(\underline{I}_{exc},\overline{I}_{exc})$ for which (C1)-(C3) hold, although at this point we have not yet determined whether $\overline{I}_{exc}$ is finite.

In fact, as we continue to increase $I$ from $\underline{I}_{exc}$, we obtain a finite value of $\overline{I}_{exc}$ at the minimal $I$ value where $(v_i(t_i),w_i(t_i))$ intersects ${\cal F}(I)$, such that (C1) fails; see Figure \ref{fig:intersect}.   The fact that (C1) fails before (C3) follows from the argument in the preceding paragraph, since $(v_i(t_i),w_i(t_i)) \in \Phi_i$ and hence $w_i(t_i) < w^*$.


The remaining issue to check is whether for any $I_{exc}$ in our candidate range $(\underline{I}_{exc},\overline{I}_{exc})$  we can establish that there are an actual timing and duration of excitation application that will allow spiking to occur.  
Consider a value $I \in (\underline{I}_{exc},\overline{I}_{exc})$, such that ${\cal C}$ intersects ${\cal F}(I)$ in two points, call them $p_1, p_2$, neither of which is $(v_i(t_i),w_i(t_i))$ or $(v^*,w^*)$.  Note that since (C1) holds for this $I_{exc}$, both $p_1, p_2$ belong to $\Phi_{0,\infty}$, not to $\Phi_i$.  
Let $t_1, t_2$ with $0 < t_1 < t_2$ denote the times it takes for the trajectory of (\ref{eq:model}) with $I=0$ and initial condition $(v_i(t_i),w_i(t_i))$ to reach $p_1, p_2$, respectively.  
If we take $I_{exc}=I$ and $t_e \in (t_1,t_2)$, then there exists a $\delta>0$ sufficiently large such that PIF occurs.  That is, the points $p_1, p_2$ lie on ${\cal F}(I)$ and hence trajectories of (\ref{eq:model})  from these points with this $I$ value blow up as $t \to \infty$.  All of the points between $p_1$ and $p_2$ on $\cal{C}$  lie below ${\cal F}(I)$ and hence yield finite-time blow-up for this $I$.
So we can choose $\underline{t}_e(I) = t_1, \overline{t}_e(I)=t_2$ and find $\delta(I,t)$ for each $t \in (\underline{t}_e(I),\overline{t}_e(I))$ to achieve PIF, as desired.    
\end{proof}

\medskip

\noindent {\bf Remark:}   Our proof shows directly that there are open intervals of $t_{e}, I_{exc}$ such that for each choice in this interval, there exists a $\delta(t_e,I_{exc})$ for which PIF occurs. Of course, this $\delta$ is not unique, since for  any larger $\delta$, the application of excitation $I_{exc}$ on the time interval $(t_e,t_e+\delta)$ will also yield a spike.  

\medskip

In some cases, it may be of interest to consider what happens when the duration $\delta$ of excitation is fixed.
Suppose we select $I_{exc}$ for which (C1) and (C2) hold and pick a point on ${\cal C}$ that lies below ${\cal F}(I_{exc})$.  Denote the time of flow from $(v_i(t_i),w_i(t_i))$ to that point under system (\ref{eq:model}) with $I=0$ by $t_{off}$. 
 If we set $I=I_{exc}$ and solve (\ref{eq:model}) with that point as the initial condition, how long do we have to wait before setting $I=0$ in order to ensure that the trajectory will yield a spike even with the $I=0$ reset?
 Geometrically, the requirement is that if we solve system (\ref{eq:model}) with $I=I(t)$, with $t_e=t_{off}$ and some $\delta$, we have that the solution at time $t_e + \delta$ lies below the threshold curve ${\cal F}(0)$, defined for $I=0$.  
 
To make this idea precise, recall that $\Phi_0(t)$ denotes the trajectory of system (\ref{eq:model}) with $I=0$ and with initial condition $(v_i(t_i),w_i(t_i))$; we will now follow this trajectory from time $t=0$ to time $t=t_e$.  
 Let $\Phi_{I_{exc}}(t)$ denote the trajectory of system (\ref{eq:model}) with $I=I_{exc}$ and initial condition $\Phi_{I_{exc}}(0)=\Phi_0(t_e)$.  Our construction gives us the following result.
 
 \medskip
 
\begin{proposition} PIF occurs for this choice of $I_{exc}, t_e$ and fixed $\delta$ if and only if $\Phi_{I_{exc}}(\delta)$ lies below ${\cal F}(0)$.
\end{proposition}

Next, suppose that $\delta>0$ is fixed and that for $I_{exc}=I$ for some choice of $I>0$, we let ${\cal P}_I(\delta)$ denote the points on ${\cal C}$ for which there is a choice of $t_e$ such that PIF occurs with that fixed $\delta$.  A final result on PIF is:

\begin{proposition}  Fix $\delta > 0$.  If $(v_1,w_1) \in {\cal P}_I(\delta)$ with $v_1>cw^-/b$ 
and $\tilde{I} > I$, then $(v_1,w_1) \in {\cal P}_{\tilde{I}}(\delta)$ as long as $(v_i(t_i),w_i(t_i))$ lies above ${\cal F}(\tilde{I})$, as needed for (C1).
\end{proposition}

\begin{proof}
We need to establish that the time of passage for trajectories of (\ref{eq:model}) with fixed $I$ from $(v_1,w_1)$ to ${\cal F}(0)$ decreases as $I$ increases.
 But this is easy to show because initially, $dv/dt$ increases in $I$ while $dw/dt$, which is positive since $v>cw^-/b$ is independent of $I$.  Hence, for sufficiently small time, the trajectory with $I$ lies above that with $\tilde{I}$.  The trajectory from $(v_1,w_1)$ with $I$ subsequently remains bounded below by that with $\tilde{I}$, since if they were to meet again, the same reasoning would apply.  Thus, the time of passage from $(v_1,w_1)$ to any fixed section $\{ v = \bar{v} > v_1 \}$  is shorter with $\tilde{I}$ than with $I$.  Finally, $W(v,0)$ increases in $v$, so that the trajectory with $I$ must reach a larger $v$-value than that with $\tilde{I}$ in order to cross the $I=0$ firing threshold curve ${\cal F}(0)$.  Thus, the time of passage from $(v_1,w_1)$ to ${\cal F}(0)$  is longer with $I$ than with $\tilde{I}$, as desired.  
  \end{proof}

  \section{Slope detection}\label{sec:SlopeDetection}
  	Another property associated with type III excitability is the  phenomenon of \emph{slope detection} that arises in response to continuous time-varying input. A neuron display slope detection if it does not fire to too rapid too slow inputs, but selectively fires to inputs with rates of change in a specific range \cite{meng2012}. To study slope detection, we consider here a simple family of \emph{tent} input:
	\begin{equation}
	\label{eq:tent}
	I_A^{\beta}(t)=\begin{cases}
		\beta t & 0\leq t \leq T_A^{\beta},\\
		\beta (2T_A^{\beta}-t) & T_A^{\beta}\leq t \leq 2T_A^{\beta},\\
		0 & \text{otherwise},
	\end{cases}
	\end{equation}
	where $T_A^{\beta}=A/\beta$. These stimuli are particularly attractive from a mathematical viewpoint because their piecewise-linear nature is convenient for analytical developments, and also their simple form allows independent variation of the amplitude (parameter $A$) and slope (parameter $\beta$). We note that because of the choice of fixing amplitude independently of the slope, the total amount of current injected defined as the integral of the instantaneous current $I_A^{\beta}$, is inversely proportional to the slope:
	\[\int_0^{2T_A^{\beta}} I_A^{\beta}(s)\,ds = A\,T_A^{\beta} = \frac{A^2}{\beta},\] 
	and thus, a non-monotone dependence of response on input slope here can be rephrased as a non-monotone dependence on total injected current.
	
	\begin{figure}[!h]
		\centerline{\includegraphics[width=\textwidth]{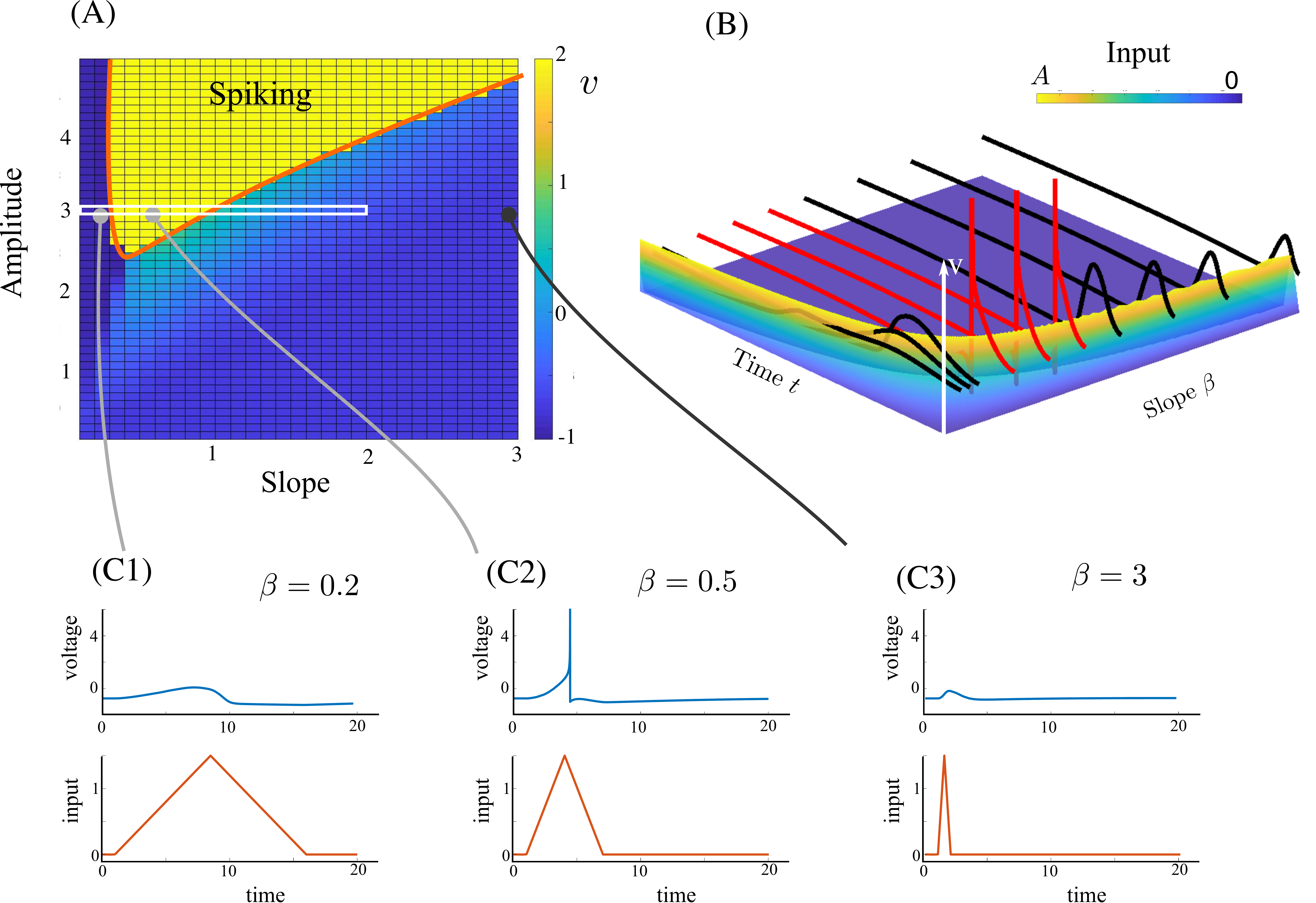}}
		\caption{Slope detection in the quartic model $F(v)=(v-\gamma/b)^4+2(v-\gamma/b)$ with $a=1$, $b=0.4$, $\gamma=0.3$ and $c=0$. (A) Starting from rest, responses to the tent input either spike (yellow) or remain subthreshold (colormap, showing the final voltage value at the end of the stimulation) as a function of stimulus slope and amplitude. The orange line delineates trajectories that fire an action potential from those that do not. For sufficiently large input amplitude, the system displays slope detection, with no spike for too low or too high slope (C1, C3 and black trajectories in B) and spiking for intermediate slope (A2 and red trajectories in B).}
		\label{fig:slope}
	\end{figure}

	Numerical simulations displayed in Fig.~\ref{fig:slope} show evidence of slope detection in an example from our class of models of interest. For fixed, sufficiently large stimulus amplitude, we observe a triphasic response as a function of the stimulus slope, with subthreshold responses for small enough or large enough slope and spiking for intermediate slopes, as visible in the examples depicted in Fig.~\ref{fig:slope}B, and, for three fixed values of slope, in panel  Fig.~\ref{fig:slope}C. Heuristically, the mechanism of slope detection arises through the conjunction of two elements: (i) the ability of the neuron to remain in the vicinity of a stable fixed point and not initiate a spike during the application of an input, and (ii) the timescale of spike initiation compared to the input slope. 
	\begin{enumerate}
		\item For sufficiently small slope, the stimulation acts as a slowly varying dynamical bifurcation parameter. In regimes associated with type III excitability, a stable equilibrium persists for a wide range of constant input levels $I$ (see Fig.~\ref{fig:heuristic}A, black line). When $I$ is varied slowly compared to the typical relaxation time towards the stable equilibrium, trajectories that start near the fixed point will remain in the vicinity of the $I$-dependent fixed point for all times, and therefore will arrive back near the resting state associated with no input at the end of the transient stimulation (see e.g., in Fig.~\ref{fig:heuristic}A, the darkest blue curves more closely following the black line of fixed points). Because this resting state is stable, one can find small enough slopes such that the state of the neuron after the end of stimulation lies in the attraction basin of that fixed point, resulting in a return to rest without a spike.
		\item For intermediate to large slopes, the stimulus may rise too quickly to allow the trajectory to closely track the fixed point and the neuron starts initiating a spike. This is visible in Fig.~\ref{fig:heuristic}A and B, starting with the trajectory highlighted with a black arrow. Two cases therefore arise depending on the decay rate of the stimulus relative to the spiking dynamics:
		\begin{enumerate}
			\item the spike fully unfolds when the slope is not sufficiently large for the input to reach its maximum and decay back down before the spike is fired (trajectories terminated with an arrow in Fig.~\ref{fig:heuristic}A,B);  
			\item alternatively, for large enough slopes, the stimulus may decay fast enough to low input levels and capture the trajectory back in the attraction bassin of the fixed point before the spike can materialize. Actually, larger slopes correspond to shorter durations of input and smaller total input, leaving the neuron voltage less affected and thus trajectories closer to the constant trajectory equal to the resting potential (dashed line in Fig.~\ref{fig:heuristic}B). 
		\end{enumerate}
	\end{enumerate}

	\begin{figure}[!h]
		\centerline{\includegraphics[width=\textwidth]{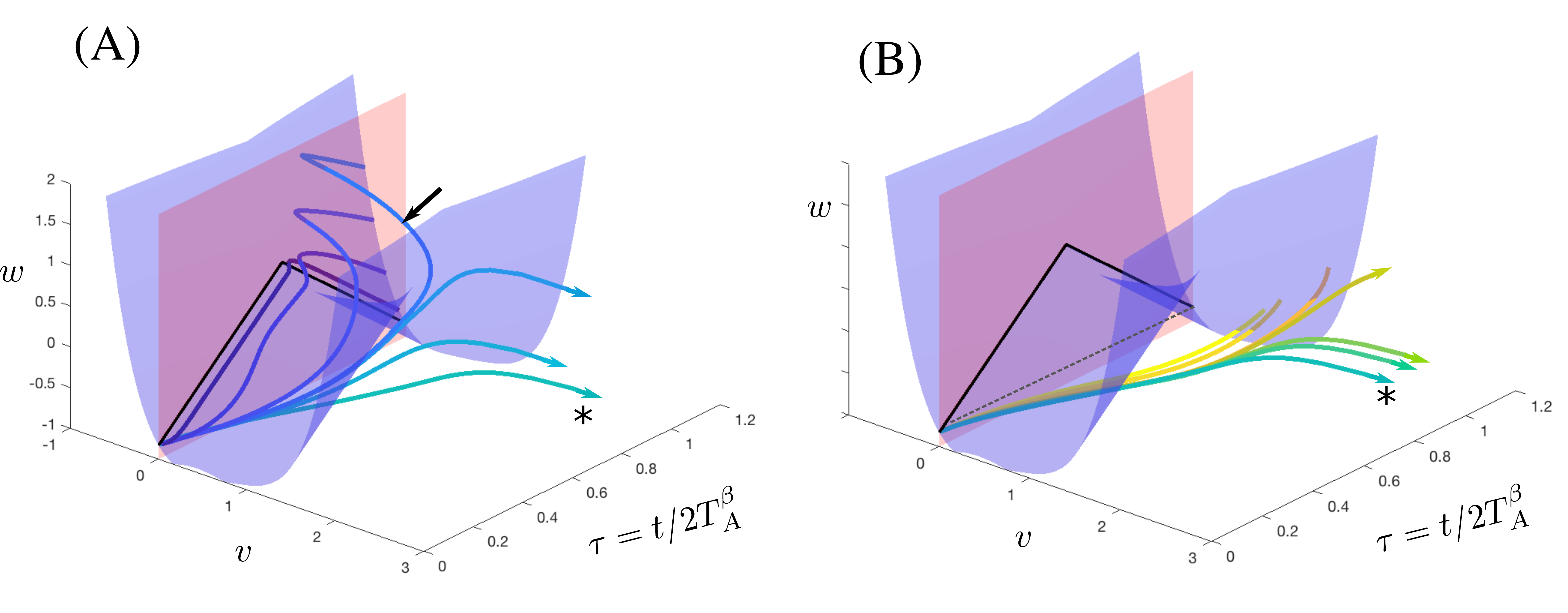}}
		\caption{Trajectories in the 3-dimensional space $(v,w,{t/2T_A^{\beta}})$ for various slopes. The time rescaling chosen makes input duration independent of slope $\beta$, allowing for comparison of trajectories. Blue surface: $v$-nullcline; red plane: $w$-nullcline; black line: fixed points plotted for input with value $I_A^{\beta}(2T_A^{\beta} \tau)$ (independent of $\beta$ as indicated). Three-dimensional curves represent trajectories; color encodes input slope. (A): low slopes (from blue to green, 0.03, 0.1, 0.2, 0.25, 0.254, 0.3, 0.6). Note how trajectories for small slopes (darkest blue) closely follow the curve of fixed points. Black arrow highlights trajectory getting transiently outside the attraction basin of the fixed point ($\beta=0.25$) but not spiking. For slope slightly larger ($\beta=0.254$), spiking arises (arrowheaded trajectories). Trajectory $\beta=0.6$ (starred) is present in both panels, A and B. (B): Starting from $\beta=0.6$ (starred), the slope is increased to $\beta=2$ and spiking is lost for some slope between $\beta=0.94$ (yellow trajectory with arrowhead) and $0.97$ (orange trajectory without arrowhead). Note how larger slopes  correspond to straighter trajectories that approaching a constant trajectory (dashed gray).}
		\label{fig:heuristic}
	\end{figure}

As indicated in the introduction, we shall relax here the assumption of strict type III excitability and allow $c>0$. We will need to ensure, however, that the input applied remains  below a maximal value $I_M$ for which the fixed point persists and remains stable. This occurs when the applied current remains below the saddle-node and Hopf bifurcation lines, which, following~\cite{touboul:08}, provides the condition:
\[I_M(b,c)=\min\left(J(b/c,b/c),J(b/c,c) \right)\]
with 
\[J(x,y)=x\left[F'\right]^{-1}(y)-F\left(\left[F'\right]^{-1}(y)\right).\]
By convention, we define $I_M(b,0)=\infty$ for any $b>0$. For any $I<I_M$, the system has a stable fixed point $(v^*(I),w^*(I))$. Furthermore, we define $I_0(b,c)=\frac{b}{c}v_f-F(v_f)$ where $v_f$ is the point where $F$ reaches its minimum ($v_f=(F')^{-1}(0)$); see assumption (A2). For $I<I_0$, the stable fixed point is also such that $F'(v^*(I))<0$.

The heuristic description of slope detection can be formalized as follows: 

\begin{theorem}\label{thm:Slope}
	Consider the solution of the bidimensional integrate-and-fire model~\eqref{eq:model} with tent input $I(t)=I_A^{\beta}(t)$ given by (\ref{eq:tent}). 
	\renewcommand{\theenumi}{(\roman{enumi})}
	\begin{enumerate}
		\item For any fixed amplitude {$A<I_M(b,c)$}, the solution is defined for all times when slope $\beta$ is large enough. 
		\item For any fixed amplitude {$A<\min(I_M(b,c),I_0(b,c))$}, the solution is defined for all time when $\beta$ is small enough. 
		\item Eventually, for any $A>A_0$ for some $A_0>0$, there exists a non-empty set of slope values for which the solutions blow up in finite time. 
	\end{enumerate}
\end{theorem}
Note that requesting $I<I_0$ for point (ii) is \emph{not} a necessary condition. Indeed, when $c>0$, the property will remain true for some choices of slopes and $A>I_0$, and another, sharper inequality can be found. Refining the boundary for $A$ would require a substantially more complex proof, which may makes the argument more obscure. Moreover, in the case where $c$ is small, refining the boundary would only allow a marginal extension of $I$ above $I_0$  (since, for stability, one needs to ensure $F'(v^*(I))<c$). Hence, we simply consider $I<I_0$ in our proof and statement. 

For theorem \ref{thm:Slope} to truly give slope detection, it is necessary that $A_0 < \min(I_M(b,c),I_0(b,c))$.  As long as this relation holds, the theorem formally expresses the slope detection property, since blow-up in the voltage variable in eq.~\eqref{eq:model} corresponds exactly to spiking of the  model neuron. Note that the amplitude $A_0$ is typically much smaller than $I_M$ and $I_0$ when $c$ is small compared to $b$ (the upper-bounds in Theorem~\ref{thm:Slope}(i,ii) actually both become trivial -- i.e., the righthand side becomes infinity -- when $c=0$).

\begin{proof}	
	The proof proceeds in three steps following the heuristic description of slope detection given above. Throughout the proof, we use $(v(t),w(t))$ to denote the solution to equation~\eqref{eq:model} with  $I_A^{\beta}(t)$ and initial condition $(v^*,w^*)$ (i.e., the fixed point of~\eqref{eq:model} associated with $I=0$). To prove that there is no spiking for large or small enough $\beta$, we will consider whether the orbit at the end of the stimulation (i.e. $(v(2T_A^\beta),w(2T_A^\beta))$) belongs to the attraction basin of $(v^*,w^*)$.  
	
	\emph{Step 1: No spiking for sufficiently large slope.} To show the absence of blow up for large slope, we show that for $\beta$ large enough, the whole orbit remains arbitrarily close to the fixed point $(v^*,w^*)$ throughout the stimulation period (Fig.~\ref{fig:heuristic} (B)). To this purpose, we denote by $U$ an open interval containing $v^*$ and on which $F$ is $k$-Lipschitz continuous for some $k \geq \max(1,b,c)$. 
	
	Define a time $\tau$ by 
	\[\tau=\begin{cases}
		2T_A^\beta & \text{if for all } t\in [0,2T_A^{\beta}], \; v(t)\in U;\\ 
		\sup\{t\leq 2T_A^\beta, v(t) \in U\} & \text{otherwise. }
		\end{cases}
		\]
Let $\vert z\vert=\vert z_1\vert +\vert z_2\vert$ for $z=\binom{z_1}{z_2}$. Keeping in mind that the fixed point $(v^*,w^*)$ is an orbit of the dynamical system, we have for any $t\leq \tau$:
	\begin{align*}
		\left\vert \binom{v(t)}{w(t)}-\binom{v^*}{w^*}\right\vert &\leq \int_0^t \left\vert \binom{F(v(s))-F(v^*) - (w(s)-w^*) + I_A^{\beta}(s)}{ b(v(s)-v^*) - c (w(s)-w^*)}\right\vert\,ds\\
	    & \leq 2k\int_0^t \left\vert \binom{v(s)}{w(s)}-\binom{v^*}{w^*}\right\vert\,ds + \int_0^t\vert I_A^{\beta}(s)\vert\,ds\\
		& \leq 2k\int_0^t \left\vert \binom{v(s)}{w(s)}-\binom{v^*}{w^*}\right\vert\,ds + \frac{A^2}{\beta}
	\end{align*}
	and, by Gronwall's lemma,
	\[\left\vert \binom{v(t)}{w(t)}-\binom{v^*}{w^*}\right\vert \leq \frac{A^2}{\beta}e^{2kt}\leq \frac{A^2}{\beta}e^{\frac{4kA}{\beta}} \]
	since $\tau \leq 2T_A^\beta$ by definition. Since the upper bound goes to $0$ as $\beta \to \infty$, we have that $\tau=2T_A^{\beta}$ for $\beta$ large enough, and moreover, at the end of the stimulus application, the orbit is arbitrarily close to $(v^*,w^*)$. Because this is an attractive fixed point of the system in the absence of applied input, for large enough slopes the endpoint of the orbit belongs to the attraction basin of $(v^*,w^*)$, implying convergence of the orbit to that point, and, in particular, the absence of blow-up.
	
		\emph{Step 2: No spiking for sufficiently small slope.} For small slopes, the input acts as a slowly varying parameter, and the orbit will track the fixed point of system~\eqref{eq:model} as a function of the input (Fig.~\ref{fig:heuristic} (A)). Under the standing assumptions $I<I_M(b,c)$ and $b> 0$, there exists a stable fixed point satisfying the equations:
	\[\begin{cases}
		v^*(I)=\frac{c}{b} w^*(I), \\
		w^*(I) = F(v^*(I))+I.
	\end{cases}\]
	For small slopes $\beta$, we will show that the orbits remain for all times close to $(v^*(I(t)),w^*(I(t)))$ for $I(t)=I_A^{\beta}(t)$. Formally, denoting
\[(v(t),w(t)) = (v^*(I(t)) + v_1(t),w^*(I(t))+ w_1(t)),\]
it is easy to show that
	\begin{equation}\label{eq:perturbation}
		\begin{cases}
		\dot{v_1}= F'(v^*(I(t)))v_1 -w_1 -(v^*)'(I(t))I'(t)+v_1^2 \eta(v_1)\\
		\dot{w_1}= bv_1-cw_1-(w^*)'(I(t)) I'(t),
		\end{cases}
	\end{equation}
	with
	\begin{equation}\label{eq:vprime}\begin{cases}
		(v^*)'(I) = c/[b-cF'(v^*(I))],\\
		(w^*)'(I) = b/[b-cF'(v^*(I))].
	\end{cases}\end{equation}
	We observe that when $c=0$, the equations~\eqref{eq:perturbation}, to leading order, simplify to a linear equation with constant coefficients and piecewise constant input, while in the case $c>0$ the equation has an additional non-homogeneous 
	part. We shall treat these cases separately, and prove that $(v_1,w_1)$ are arbitrarily small when $\beta$ is small using (i) a Lyapunov function-type argument for $c>0$ and (ii) the variation of constants formula for $c=0$. 
	
\underline{\textbf{Case $c>0$}}	When $c>0$, because the linear term $F'(v^*(I(t))$ depends on time, equation~\eqref{eq:perturbation} is a differential equation with non-constant coefficients, with non-homogeneous input proportional to the slope coefficient. Even in the absence of that input, it is well-known that that the fixed point $0$ may not be stable although the linear part may have strictly negative eigenvalues for all times. Here, the situation is, however, favorable and the trajectories remain arbitrarily close to $0$. Indeed, letting $N(t)=\frac 1 2 (v^2_1(t)+b^{-1}w^2_1(t))$, we have:
	\[\dot{N} = F'(v^*(I(t))) v_1^2 - \frac{c}{b} w_1^2+v_1^3\eta(v_1)-\Big[(v^*)'(I(t))v_1+b^{-1} (w^*)'(I(t))w_1\Big]I'(t)\]
	and letting $\delta=\min(\inf_{I\in [0,A]} \vert F'(v^*(I))\vert, c 
	)$ (the assumption $I<I_0(b,c)$ ensures $\delta>0$), we find:
	\[\dot{N} \leq -\delta N(t)+v_1^3\eta(v_1)-\Big[(v^*)'(I(t))v_1+b^{-1} (w^*)'(I(t))w_1\Big]I'(t).\]
	Fix $R>0$ smaller than 1. Until the first exit time of the trajectory $(v_1(t),w_1(t))$ from the closed ellipse $B_R^2$ defined by
	\[B_R^2=\left\{(v,w) \in \R^2\;;\; \frac 1 2 \left(v^2+\frac{1}{b} w^2\right)\leq R\right\},\]
	 we find, by applying upper bounds and integrating, that
	\[N(t) \leq M_1 R^3 + M_2 \beta\]
	where $M_1=\delta^{-1}\sup_{\vert x\vert <1, I\in [0,A]} \vert \eta(x,I)\vert $ and $M_2=\delta^{-1}\sup_{I\in[0,A]} \vert (v^*)'+b^{-1}(w^*)'\vert $, which is finite under our assumptions (the denominator in~\eqref{eq:vprime} is the determinant of the Jacobian, which remains away from 0). Taking $R$ and $\beta$ small enough completes the proof. Explicitly, taking $R\leq \min(1,(M_1+M_2)^{-1/2})$ and $\beta<R^3$, we obtain:
	\[\frac{M_1 R^3 + M_2\beta}{R} \leq (M_1+M_2)R^2 \leq 1,\]
	implying that solution $(v_1,w_1)$ remains in $B_R^2$ for all times. In particular, choosing $R^\star$ such that 
	the attraction basin of $(v^*,w^*)$ contains the open ball $(v^*,w^*)+B^2_{R^\star}$, we can find $\beta$ small enough such that the trajectory  $(v(t),w(t))$ remains trapped in a tube $\tilde{B}^2_{R^\star}(t)=\{(v,w) : \ \frac{1}{2}((v-v^*(I(t)))^2+\frac{1}{b}(w(t)-w^*(I(t)))^2 \leq R^\star\}$ around $(v^*(I(t)),w^*(I(t)))$ for all $t\in [0, 2T_A^{\beta}]$, and therefore, at the end of the stimulation, will converge to the fixed point $(v^*,w^*)$. 
	
	
 
\underline{\textbf{Case $c=0$}} When $c=0$, the leading order part of the system simplifies into a linear equation with constant coefficients and piecewise-constant inhomogeneity:
	\begin{equation}\label{eq:perturbation2}
	\begin{cases}
		\dot{v_1}=-\nu v_1 -w_1 +v_1^2\eta(v_1)\\
		\dot{w_1}= b v_1 -I'(t)
	\end{cases}
	\end{equation}
	with $\nu=\vert F'(v^*)\vert=-F'(0)$. The proof used for $c>0$ fails here (since $\delta$ becomes equal to $0$), and we prove an analogous result based on the variation of constants formula. Writing the system in the coordinates given by the eigenvectors of the matrix $M$ associated with the linear part of system (\ref{eq:perturbation2}), we get
\[
\begin{cases}
	\dot{x}=\lambda_1 x + f(x,y) + \kappa(t) \mu_1 \\
	\dot{y}=\lambda_2 y + g(x,y) + \kappa(t) \mu_2
\end{cases}
\] 
where $\lambda_{1,2}$ are the eigenvalues of  $M$,
$\mu_{1,2}$ are complex numbers only depending on the eigenvector coefficients, and $f$ and $g$ are the projections of the term $v_1^2\eta(v_1)$ on the eigenvectors, rewritten in terms of the $(x,y)$ coordinates. Therefore, both functions are $O(\vert x\vert^2+\vert y\vert^2)$. Using the variation of constants formula, we compute, e.g., for the coordinate $x$, the inequality
\[\vert x(t)\vert \leq \int_0^t e^{\lambda_1^R(t-s)/2} \vert f(x(s),y(s))\vert \,ds + \frac{2\vert \beta \mu_1\vert}{\lambda_1^R}\]
where $\lambda_1^R<0$ is the real part of $\lambda_1$. An analogous formula is valid for $y$. The proof is completed by showing that the bound on $|x(t)|$ can be made arbitrarily small for $\beta$ small enough. To this purpose, define by $B_R^{1}$ the open ball of radius $R>0$ for the norm $\Vert (x,y)\Vert_1 = \vert x \vert + \vert y \vert$. For $R$ small, we have $\max_{B^1_R}(f(x,y),g(x,y)) \leq C R^2$ for some fixed constant $C$. Therefore, before the solution leaves $B_R^1$, we have:
\[\Vert (x(t),y(t))\Vert_1 \leq \frac{2 [C R^2 +\beta (\vert \mu_1\vert+\vert \mu_2\vert)]}{\min(\lambda_1^R,\lambda_2^R)}.\]
We conclude using a method analogous to the end of the argument for $c>0$, by choosing $R$ and $\beta$ small enough to ensure that the trajectory remains in $B_R^1$ for all times, and then take $R$ (and $\beta$ accordingly) small enough to ensure that the orbit belongs to the attraction basin of $(v^*, w^*)$ at the end of stimulation.



	

\emph{Step 3: Spiking for intermediate slope and sufficiently large amplitude.} We prove that, for fixed stimulation slope $\beta>0$, there exists some input amplitude for which the neuron fires. To this end, we shall show that there exists a linear function $B(A)$ (with non-zero slope) and a minimal amplitude $A_0$ for which the neuron with tent input of amplitude $A>A_0$ and slope $\beta=B(A)$ blows up in finite time. Thus, if we fix $\beta>0$ in the range of $B(A)$, spiking will occur for $A > B^{-1}(\beta)$, which readily implies statement (iii) of the theorem.

First,  we consider the dynamics of \eqref{eq:model} with $I=0$.  We notice that there exists a domain $\Gamma=\{(v,w); v> \alpha w+\zeta\}$ for some $\alpha>0$ and $\zeta \in \R$  that is forward invariant under the flow, such that any neuron with initial condition within $\Gamma$ blows up in finite time\footnote{In other words, the spiking threshold is to the left of the line $v=\alpha w+\zeta$.}. Indeed, for any $\alpha>0$ such that $(\alpha b + \frac{1}{\alpha}-c)>0$, we can find $\zeta$ such that for any $v\in\mathbb{R}$,
\begin{equation}
\label{eq:Fineq}
F(v) - (\alpha b + \frac{1}{\alpha}-c)v+(\frac 1 \alpha -c)\zeta >0,
\end{equation}
because $F$ is a convex function with $\lim_{v\to-\infty} F'(v)\leq 0$, implying that $F(v) - \mu v$ is a convex, lower-bounded function for $\mu>0$. Inequality (\ref{eq:Fineq}) being satisfied for all $v$ implies that the derivative with respect to time $t$ of $v(t)-\alpha w(t)-\zeta$ on the boundary $v=\alpha w + \zeta$ is strictly positive, since 
\[\dot{v}-\alpha \dot{w}=F(v) - (\alpha b + \frac{1}{\alpha}-c)v+(\frac 1 \alpha -c)\zeta>0,\]
which ensures that $\Gamma$ is a forward flow-invariant domain. Moreover, any trajectory entering $\Gamma$ will blow up in finite time, since in that region the voltage is a strictly increasing function, and $\dot{v}\geq F(v)- (v-\zeta)/\alpha 
=:g(v)$ with $1/g(v)$ integrable at infinity, classically implying blow-up of the solution. 


Next, we consider the stimulated system during the rising phase of the input, $t\leq T_A^{\beta}$. During that phase, $w$ is increasing and remains below the fixed point associated with constant input $I=A$, which we denote by $w_A = F(v^*(A))+A$. Let us in addition assume that $(v(t),w(t))$ belongs to $\Gamma^c$ during the stimulation. In that case, we have $v(t)\leq \alpha w_A+\zeta =: v_A$. When $c=0$, this implies that
\[\dot{w}=bv\leq bv_A,\]
and hence
\[w(t)\leq w^* + b\,v_A t.\]
For $c>0$, Gronwall's lemma (or direct solution of the $w$ equation with $v=v_A$) implies that:
\[w(t)\leq w^* e^{-ct}+\frac{b\,v_A}{c}(1-e^{-ct}).\]
Noting that $(1-e^{-ct}) \leq c\,t$ for $c>0$, we obtain a common inequality for both cases $c=0$ and $c>0$:
\[w(t)\leq \vert w^* \vert +b\,v_A t.\] 
Using this inequality on $w$, we can derive the following differential inequality for $v$:
\[\dot v \geq F(v)-\vert w^* \vert + (\beta - b\,v_A)t.\]
Until now, the slope $\beta$ has been arbitrary; we now choose $\beta=2\,\delta_A$ with $\delta_A=b\,v_A$, which is the line in the plane $(A,\beta)$ on which we will show blow up of the solution for sufficiently large $A$\footnote{Note that with this relationship between $A$ and $\beta$, the duration of the stimulus is now bounded even when $A$ increases (since slope increases as well), and saturates at a value $\frac{1}{b\alpha}$.}. We thus obtain the inequality 
\[\dot v \geq F(v)-\vert w^* \vert + \delta_A\,t.\]

We now consider the differential equation
\[\dot x = F(x)-\vert w^*\vert + \delta_A\,t\]
with initial condition $x(0)=v^*$, and let $v_{mid}$ denote the value of the solution halfway through the increasing part of the tent stimulus, i.e. $v_{mid}= x(t_0)$ with  $t_0=\frac 1 2 T_A^\beta = \frac{A}{4\delta_A}$. 
For $t\in [\frac 1 2 T_A^\beta,\frac 3 2 T_A^\beta]$, we have 
\[\dot x \geq F(x)-\vert w^*\vert + \delta_A\,t_0 =  F(x)-\vert w^* \vert+\frac{A}{4}.\]
For $A$ large enough, the solution to this equation blows up in finite time for all initial conditions.  Moreover, since the blow up time arises before the time
\[t_0+\int_{v_{mid}}^{\infty}\frac{dv}{F(v)-\vert w^*\vert+A/4}, \] 
which can be made arbitrarily close to $t_0$ as $A$ diverges, the solution can blow up arbitrarily fast by a simple application of the monotone convergence theorem (noting that $A\mapsto v_{mid}$ is increasing and that $1/F(v)$ is integrable). This in particular implies that $v$ will necessarily enter the set $\Gamma$ before $T_A^\beta$ by choosing $A$ large enough, contradicting our assumption. 

Therefore, the trajectory of $(v,w)$ enters $\Gamma$ before time $\frac{3}{2} T_A^\beta$. We have seen that $\Gamma$ is a trapping region for the system with no input. It remains a trapping region for the system with positive input, since for $I(t)\geq 0$ we have on the boundary of $\Gamma$ that 
\[\dot{v}-\alpha \dot{w} = F(v)-w + I(t)-\alpha b v +\alpha c w \geq F(v) - (\alpha b + \frac{1}{\alpha}-c)v+(\frac 1 \alpha -c)\zeta>0.\]
The trajectory is thus trapped in $\Gamma$, where $v$ blows up in finite time before time $\frac 3 2 T_A^\beta$. 

We conclude that for any $\alpha$, there exists $A^*(\alpha)$ large enough such that the neuron spikes for any $A>A^*(\alpha)$ and for $\beta=2b\,v_A = 2b[ \alpha w^*(A) + \zeta]$ for any choice of $(\alpha,\zeta)$ for which (\ref{eq:Fineq}) holds for all $v$. 

In conclusion, for a fixed $A> A^*(\alpha)$, there exists a non-empty set of values of the slope $\beta$ for which the neuron blows up in finite time, completing the proof.
\end{proof}

\section{Discussion}
\label{sec:Discussion}

Hodgkin classified neurons based on their responses to sustained injected currents, with Type I and Type II neurons firing repeatedly to current injection, albeit with differing f-I relations, and Type III neurons giving only a phasic or transient response before returning to quiescence \cite{hodgkin1948}.  Subsequent research has fleshed out this picture, establishing associations between Hodgkin's neuron types and various additional dynamic properties and mathematical features.  As a part of this development, Type III neurons have generally been considered as those that have a stable critical point at a  resting voltage for all levels of injected current.
In this paper, we have  shown analytically and numerically that two of the major properties commonly associated with Type III neurons, {\em post-inhibitory facilitation (PIF)} and {\em slope detection}, are in fact always present under conditions that give rise to phasic responses to sustained currents.  Moreover, the ubiquity of these properties extends to cases where the stable critical point is lost as injected current increases, although it should be maintained for input levels that are directly involved in these phenomena.  We have proven these results in a well-established, fairly general two-dimensional hybrid model that combines continuous dynamics with a discrete reset and that has been shown previously to generate a wide range of dynamics reminiscent of neuronal activity \cite{touboul:08,touboul:09,rubinDCDSI,rubinDCDSII}.  In this model, spiking corresponds to a finite-time blow-up of the voltage variable, during which the second, adaptation variable in the model remains finite.  After a spike occurs, both variables are reset, but both PIF and slope detection are based on whether or not a spike is fired at all and do not involve dynamics subsequent to spike generation, so the reset does not factor into our analysis.  
Because of the dynamic richness of this model, the generality of the dynamic principles that appear in our analysis, and the excellent match between the PIF and slope detection that we observe and similar dynamics in conductance-based models \cite{dodla,meng2012}, we strongly expect that our findings apply to neural models broadly beyond the one that represents the specific focus of our study.

There has been previous theoretical work on PIF, slope detection, and Type III dynamics.  Prescott et al. extended past work by Rinzel and Ermentrout \cite{rinzel1998} by unifying Type I, II and III dynamics as different parameter regimes within minimal two- and three-dimensional models.  Like Rinzel and Ermentrout, they characterized these behaviors dynamically in terms of bifurcations induced by input currents and they also provided interpretations in terms of interacting inward and outward currents \cite{prescott2008}.  These works did not, however, go on to discuss additional input processing properties such as PIF and slope detection.  Two other papers are thus more direct progenitors of the work that we present. 
First, Dodla et al. demonstrated that PIF occurs in experiments done with neurons of the medial superior olive (MSO)  of the gerbil auditory brainstem.  They captured this phenomenon in a Hodgkin-Huxley (HH) type model of the MSO with standard sodium, potassium and leak currents as well as a low-threshold potassium current.  They also explored this effect in a reduced, planar version of the model without repolarization and in an integrate-and-fire model made planar by the inclusion of a voltage-dependent threshold \cite{dodla}. 
Their phase plane analysis computationally demonstrated the role of the spike threshold in PIF, which features in our proofs.
Second, Meng et al. also considered the HH-type MSO model and computationally illustrated slope detection, phase locking, and coincidence detection properties both in the full model and in planar reductions \cite{meng2012}.  Note that we have not considered phase locking and coincidence detection, and thus the rigorous treatment of these additional behaviors that are thought of as Type III dynamics remains for future work.  

Another viewpoint on neuronal classification is represented within the recent work of Ly and Doiron \cite{ly2017} (see also the references therein).  Building on past work in a similar vein, these authors analyze the three neuronal types in terms of the frequencies of stochastic, oscillatory stimuli to which they respond, referring to Class III neurons as those with high pass selectivity, also known as phasic neurons.  Despite this difference in perspective and the authors' emphasis on stochastic effects, this work is clearly relevant to our study.  In particular, the idea of an input-dependent spiking threshold features in their analysis.  There is also some similarity between these neurons' failure to respond to low frequency inputs, as emphasized by Ly and Doiron, and their failure to respond to inputs with sufficiently shallow slopes, as arises in slope detection, since both of these input types features a gradual ramping up and down of input strength.   Moreover, Ly and Doiron comment on the intuition for why hyperpolarization followed by depolarization can evoke spike responses in these neurons, which they notice from spike-triggered averaging. Our analysis of PIF and slope detection strongly suggest that these properties will be robust to noise, since small perturbations of trajectory paths are unlikely to induce threshold crossings.  As in previous work, however, the inclusion of stochasticity in voltage dynamics would likely smear out these effects, since noise will make spiking responses to specific input patterns probabilistic rather than all-or-none, and a careful study of these effects remains for future work.
Additional authors, notably Gai et al. \cite{gai2010} and Ratt\'{e} et al. \cite{ratte2015}, also have commented on noise-based encoding of slow signals in phasic neurons and encoding of derivatives of input based on subthreshold membrane currents, respectively.
 
Separation of timescales is a property that has arisen in some considerations of threshold crossings in responses to inputs in neuroscience contexts.   A recent example showed numerically that for a bistable two-dimensional system representing synaptic weight dynamics with a saddle separatrix between low and high weight stable critical points, a non-monotonic relationship arises between the stimulus amplitude needed to cross the separatrix and overall stimulus area needed to cross if there is timescale separation between the two variables \cite{gastaldi2019}.  In this case, when the amplitude of the input is too large, its duration is too short to allow a crossing, which resembles the large slope regime in our slope detection analysis.  Importantly, our analysis does not require timescale separation.  In fact, timescale separation may work counter to PIF, since slowing down the adaptation dynamics reduces the extent of adaptation that occurs during the return to rest after the removal of inhibition.  

Although model (\ref{eq:model}) is a hybrid system with a spike threshold and a discrete resetting condition that is applied after a spike occurs, the PIF and slope detection phenomena depend on the conditions under which a spike is fired at all, not on what happens after a spike.  Hence, our analysis is entirely based on the continuous dynamics of the model, and this model was a convenient choice for analysis simply because the lack of a continuous repolarization mechanism allows the use of a simpler vector field than would otherwise be possible.  
The mechanisms that we have shown to give rise to PIF and slope detection are clearly not specific to hybrid models and hence this work establishes that PIF and slope detection are expected to be quite general phenomena.  In fact, the neural model that we consider for $c>0$ fits the Type II classification, in the sense that its resting critical point destabilizes through an Andronov-Hopf bifurcation as input increases.  More generally, PIF requires that a system has a stable resting critical point for a range of nonnegative input values including zero, and inhibition must recruit a positive feedback or remove a negative feedback such that it becomes easier, but not automatic, to fire after the application and subsequent removal of inhibition.  Standard spiking mechanisms involve the activation of a negative feedback potassium current that repolarizes the membrane potential and that deactivates in response to hyperpolarization, so the latter requirement is commonly satisfied.  Failure to respond to inputs with small slopes also requires persistence of a stable resting critical point in the model  that we consider.   On the other hand, failure to respond to inputs with large slopes does not have this requirement; the input level can increase through a value where the rest point loses stability, as long as it drops back down again quickly enough.  Interestingly, in Type I neural models, the stable critical point is lost through a SNIC bifurcation with increasing input, but just beyond this bifurcation, the escape from the subthreshold voltage regime is extremely slow.  This slow escape could potentially extend the failure to respond to shallow inputs even to inputs reaching levels above the SNIC bifurcation, but this idea remains open for careful investigation.

Finally, as we have noted, we have not considered the additional phenomena of phase locking and coincidence detection identified in past work as aspects of Type III dynamics \cite{dodla,meng2012}.  In phase locking, a neuron exposed to an oscillatory input only activates during certain phases of that stimulus.  In coincidence detection, a neuron responds to multiple inputs only if they arrive close enough together in time.  Clearly, these properties are not exclusive to Type III neurons, but numerics have shown that relative to neurons with sustained responses to input, the firing of neurons with phasic responses to input occurs in a more limited time range (for phase locking) and requires more temporally proximal inputs  (for coincidence detection) \cite{meng2012}.  A previous study used an input-dependent spiking threshold to rigorously study coincidence detection in certain planar neuronal models featuring a voltage differential equation coupled to a second equation for the decay of the input strength \cite{rubin2006} (see also \cite{wang2011}).  Based on those results and our current work, we expect that similar geometric dynamical systems ideas should be useful for future studies of necessary and sufficient conditions for phase locking and coincidence detection in neural models.

\vskip 0.3cm

\noindent\textbf{Acknowledgements.} JR received support from NSF award DMS1612913.  JS-R  was supported by NCN (National Science Centre, Poland) grant no.~2019/35/D/ST1/02253. JT received support from NSF award DMS-1951369.

\bibliographystyle{siamplain}
\bibliography{PIF}

\end{document}